\documentclass[letterpaper,twocolumn,10pt]{article}
\usepackage{usenix2019, epsfig,endnotes}
\usepackage[utf8]{inputenc}
\usepackage{color}
\usepackage{multirow}
\usepackage{url}
\usepackage{algorithm}
\usepackage[noend]{algpseudocode}
\usepackage{subfigure}
\usepackage{enumitem} 
\usepackage{times}
\usepackage{xspace}
\usepackage[table]{xcolor}
\usepackage{booktabs}
\usepackage{latexsym}
\usepackage{amsmath}
\usepackage{amsthm}
\usepackage{amssymb}
\usepackage{tabularx}
\usepackage{balance}

\setlist{nolistsep}

\newcommand{\para}[1]{\smallskip\noindent\textbf{#1}}
\newcommand{\parait}[1]{\smallskip\noindent\textit{#1}}
\newcommand{\paraf}[1]{\noindent\textbf{#1}}
\newcommand{\cut}[1]{}

\newcommand{\COMMENTS}{no}

\newcommand{\cale}{\mathcal{E}}
\newcommand{\calG}{\mathcal{G}}
\newcommand{\calM}{\mathcal{M}}
\newcommand{\OPT}{\mathrm{opt}}
\newcommand{\CUT}{\mathrm{cut}}

\newtheorem{theorem}{Theorem}

\newtheorem{definition}{Definition}
\newtheorem{lemma}{Lemma}

\newcommand{\sysname}{DistCache\xspace}
\newcommand{\cpartition}{CachePartition\xspace}
\newcommand{\creplication}{CacheReplication\xspace}
\newcommand{\nocache}{NoCache\xspace}

\ifthenelse{\equal{\COMMENTS}{yes}}{%
    \setlength{\marginparwidth}{1.5cm}
    \usepackage[draft]{todonotes}
}{%
    \usepackage[disable]{todonotes}
}%
\newcommand{\alan}[1]{\todo[author=Alan, color=green]{\footnotesize #1}}

\begin{document}
\date{}

\title{\sysname: Provable Load Balancing for Large-Scale  Storage Systems \\
with Distributed Caching}

{
\author{Zaoxing Liu$^\star$, Zhihao Bai$^\star$, Zhenming Liu$^\dag$, Xiaozhou Li$^\bullet$, \\Changhoon Kim$^\ddag$, Vladimir Braverman$^\star$, Xin Jin$^\star$, Ion Stoica$^\diamond$
\\\small $^\star$Johns Hopkins University  $^\dag$College of William and Mary	$^\bullet$Celer Network $^\ddag$Barefoot Networks $^\diamond$UC Berkeley
}
}
\maketitle

\begin{abstract}
Load balancing is critical for distributed storage to meet strict service-level
objectives (SLOs). It has been shown that a fast cache can guarantee load
balancing for a clustered storage system. However, when the system scales out to
multiple clusters, the fast cache itself would become the bottleneck.
Traditional mechanisms like cache partition and cache replication either result
in load imbalance between cache nodes or have high overhead for cache coherence.

We present \sysname, a new distributed caching mechanism that provides provable
load balancing for large-scale storage systems. \sysname co-designs cache
allocation with cache topology and query routing. The key idea is to partition
the hot objects with independent hash functions between cache nodes in different
layers, and to adaptively route queries with the power-of-two-choices. We prove
that \sysname enables the cache throughput to increase linearly with the number
of cache nodes, by unifying techniques from expander graphs, network flows, and
queuing theory. \sysname is a general solution that can be applied to many
storage systems. We demonstrate the benefits of \sysname by providing the
design, implementation, and evaluation of the use case for emerging switch-based
caching.
\end{abstract}

\section{Introduction}
\label{sec:introduction}

Modern planetary-scale Internet services (e.g., search, social networking and
e-commerce) are powered by large-scale storage systems that span hundreds to
thousands of servers across tens to hundreds of racks~\cite{gfs, dynamo,
haystack, memcache-nsdi13}. To ensure satisfactory user experience, the storage
systems are expected to meet strict service-level objectives (SLOs), regardless
of the workload distribution. A key challenge for scaling out is to achieve load
balancing. Because real-world workloads are usually
highly-skewed~\cite{workload-fb-sigmetrics12, benchmarking-socc10,
Huang:2014:HotNets, Jung:fcd}, some nodes receive more queries than
others, causing hot spots and load imbalance. The system is bottlenecked
by the overloaded nodes, resulting in low throughput and long tail latencies.

Caching is a common mechanism to achieve load balancing~\cite{Fan:smallcache,
switchkv, netcache}.
An attractive property of caching is that caching $O(n \log n)$ hottest
objects is enough to balance $n$ storage nodes, regardless
of the query distribution~\cite{Fan:smallcache}. The cache size only
relates to the number of storage nodes, despite the number of objects stored in
the system. Such property leads to recent advancements like
SwitchKV~\cite{switchkv} and NetCache~\cite{netcache} for balancing
clustered key-value stores.

Unfortunately, the small cache solution cannot scale out to multiple clusters.
Using one cache node per cluster only provides \emph{intra-cluster} load
balancing, but not \emph{inter-cluster} load balancing. For a large-scale storage system
across many clusters, the load between clusters (where each cluster can be
treated as one ``big server'') would be imbalanced. Using another cache node,
however, is not sufficient, because the caching mechanism requires the cache to
process \emph{all} queries to the $O(n \log n)$ hottest
objects~\cite{Fan:smallcache}. In other words, the cache throughput needs to be
no smaller than the \emph{aggregate} throughput of the storage nodes.


As such, it requires another caching layer with multiple
cache nodes for inter-cluster load balancing.
The challenge is on cache allocation. Naively
replicating hot objects to all cache nodes incurs high overhead for cache
coherence. On the other hand, simply partitioning hot objects between the cache
nodes would cause the load to be imbalanced between the cache nodes. The system
throughput would still be bottlenecked by one cache node under highly-skewed
workloads. Thus, the key is to carefully partition and replicate hot objects, in
order to avoid load imbalance between the cache nodes, and to reduce the
overhead for cache coherence.

We present \sysname, a new distributed caching mechanism that provides provable
load balancing for large-scale storage systems. \sysname enables a ``one big
cache'' abstraction, i.e., an \emph{ensemble} of fast cache nodes acts as a single
ultra-fast cache. \sysname co-designs cache allocation with
multi-layer cache topology and query routing. The key idea is to use independent
hash functions to partition hot objects between the cache nodes in different layers,
and to apply the power-of-two-choices~\cite{power-of-two} to adaptively route queries.

Using independent hash functions for cache partitioning ensures that if a cache
node is overloaded in one layer, then the set of hot objects in this node would
be distributed to multiple cache nodes in another layer with high probability.
This intuition is backed up by a rigorous analysis that leverages expander
graphs and network flows, i.e., we prove that there exists a solution to split
queries between different layers so that no cache node would be overloaded in
any layer. Further, since a hot object is only replicated in each layer once, it
incurs minimal overhead for cache coherence.

Using the power-of-two-choices for query routing provides an efficient,
distributed, online solution to split the queries between the layers.
The queries are routed to the cache nodes in a distributed way based on cache
loads, without central coordination and without knowing what is the optimal
solution for query splitting upfront. We leverage queuing theory to show it is
asymptotically optimal. The major difference between our problem and the
balls-and-bins problem in the original power-of-two-choices
algorithm~\cite{power-of-two} is that our problem hashes objects into cache
nodes, and queries to the same object \emph{reuse} the same hash
functions to choose hash nodes, instead of using a \emph{new random
source} to sample two nodes for each query. We show that the power-of-two-choices
makes a ``life-or-death'' improvement in our problem, instead of a ``shaving off a
log $n$'' improvement.

\sysname is a general caching mechanism that can be applied to many storage
systems, e.g., in-memory caching for SSD-based storage like
SwitchKV~\cite{switchkv}  and switch-based caching for in-memory storage like
NetCache~\cite{netcache}. We provide a concrete system design to scale out
NetCache to demonstrate the power of \sysname. We design both
the control and data planes to realize \sysname for the emerging switch-based
caching. The controller is highly scalable as it is off the critical path. It is
only responsible for computing the cache partitions and is not involved in
handling queries. Each cache switch has a local agent that manages the hot
objects of its own partition.

The data plane design exploits the capability of
programmable switches, and makes innovative use of \emph{in-network telemetry}
beyond traditional network monitoring to realize \emph{application-level
functionalities}---disseminating the loads of cache switches by
piggybacking in packet headers, in order to aid the power-of-two-choices.
We apply
a two-phase update protocol to ensure cache
coherence.

In summary, we make the following contributions.

\begin{itemize}[leftmargin=*]
    \item We design and analyze \sysname, a new distributed caching mechanism
    that provides provable load balancing for large-scale storage systems
    (\S\ref{sec:mechanism}).

    \item We apply \sysname to a use case of emerging switch-based caching, and design a concrete
    system to scale out an in-memory storage rack to multiple racks (\S\ref{sec:design}).

    \item We implement a prototype with Barefoot Tofino switches and
    commodity servers, and integrate it with Redis (\S\ref{sec:implementation}). Experimental
    results show that \sysname scales out linearly with the number of racks, and the
    cache coherence protocol
    incurs minimal overhead (\S\ref{sec:evaluation}).
\end{itemize}

\section{Background and Motivation}
\label{sec:background}

\subsection{Small, Fast Cache for Load Balancing}
\label{sec:background:cache}

As a building block of Internet applications, it is critical
for storage systems to meet strict SLOs. Ideally, given the
per-node throughput $T$, a storage system with $n$ nodes should
guarantee a total throughput of $n \cdot T$.
However, real-world workloads are usually high-skewed, making it challenging to
guarantee performance~\cite{workload-fb-sigmetrics12, benchmarking-socc10,
Huang:2014:HotNets, Jung:fcd}. For example, a measurement study on the Memcached
deployment shows that about 60-90\% of queries go to the hottest
10\% objects~\cite{workload-fb-sigmetrics12}.

\begin{figure}[t]
\centering
\includegraphics[width=0.9\linewidth]{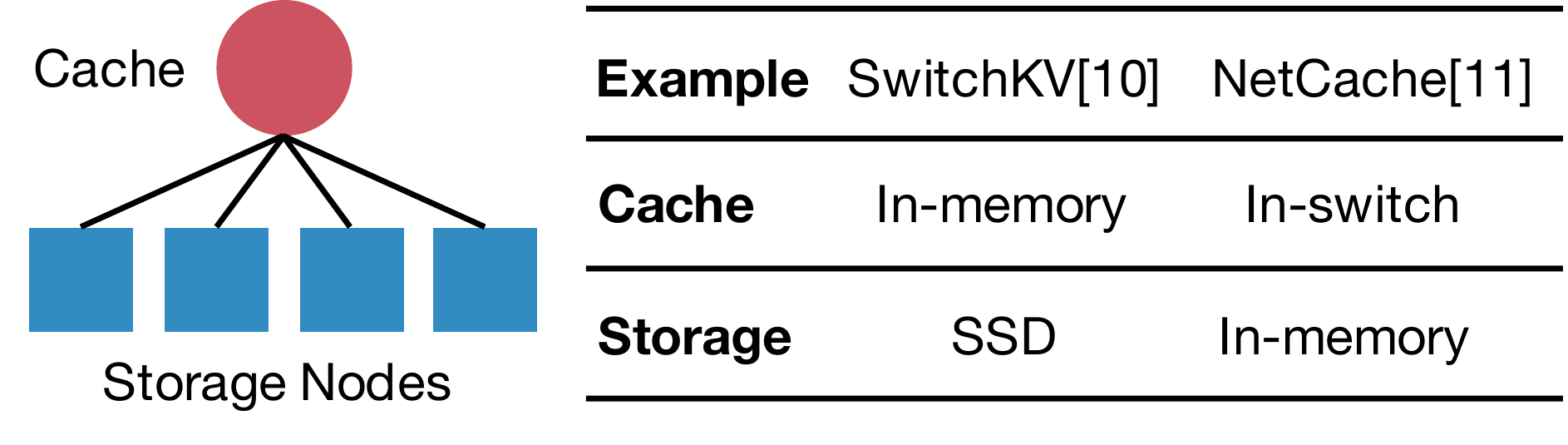}
\vspace{-0.18in}
\caption{Background on caching. If the cache node can absorb \emph{all} queries
to the hottest $O(n\log n)$ objects, the load on the storage nodes is guaranteed
to be balanced~\cite{Fan:smallcache}.}

\label{fig:Background_cache}
\end{figure}

Caching is a common mechanism to achieve load balancing for distributed storage,
as illustrated in Figure~\ref{fig:Background_cache}. Previous work has proven
that if the cache node can absorb \emph{all} queries to the hottest $O(n\log n)$
objects, then the load on $n$ storage servers is guaranteed to be balanced,
despite query distribution and the total number of objects~\cite{Fan:smallcache}.
However, it also requires that the cache throughput needs to be at least $n
\cdot T$ to not become the system bottleneck. Based on this theoretical
foundation, SwitchKV~\cite{switchkv} uses an in-memory cache to balance
SSD-based storage nodes, and NetCache~\cite{netcache} uses a switch-based cache to
balance in-memory storage nodes. Empirically, these systems have shown that
caching a few thousand objects is enough for balancing a hundred storage nodes,
even for highly-skewed workloads like Zipfian-0.9 and Zipfian-0.99~\cite{switchkv,
netcache}.

\subsection{Scaling out Distributed Storage}
\label{sec:background:scale}

The requirement on the cache performance limits the system scale. Suppose the
throughput of a cache node is $\widetilde{T} = c \cdot T$. The system can scale
to at most a cluster of $c$ storage nodes. For example, given that the typical
throughput of a switch is 10-100 times of that of a server,
NetCache~\cite{netcache} can only guarantee load balancing for 10-100 storage servers.
As such, existing solutions like SwitchKV~\cite{switchkv} and NetCache~\cite{netcache}
are constrained to one storage cluster, which is typically one or two racks of servers.

For a cloud-scale distributed storage system that spans many clusters, the load
between the clusters can become imbalanced, as shown in
Figure~\ref{fig:motivation_problem}(a). Naively, we can put another cache node in front
of all clusters to balance the load between clusters.
At first glance, this seems a nice solution,
since we can first use a cache node in each cluster for \emph{intra-cluster} load
balancing, and then use an upper-layer cache node for \emph{inter-cluster} load
balancing. However, now each cluster becomes a ``big server'', of which the
throughput is already $\widetilde{T}$. Using only one cache node cannot meet the
cache throughput requirement, which is $m\widetilde{T}$ for $m$ clusters. While using
multiple upper-layer cache nodes like Figure~\ref{fig:motivation_problem}(b) can
potentially meet this requirement, it brings the question of how to allocate hot
objects to the upper-layer cache nodes. We examine two traditional cache allocation
mechanisms.

\para{Cache partition.} A straightforward solution is to partition the object
space between the upper-layer cache nodes. Each cache node only caches the hot
objects of its own partition. In this case, a write query will update only one upper-layer cache node for cache coherence.  
Cache partition works well for uniform workloads, as the
cache throughput can grow linearly with the number of cache nodes. But remember
that under uniform workloads, the load on the storage nodes is already
balanced, obviating the need for caching in the first place. The whole purpose
of caching is to guarantee load balancing for skewed workloads. Unfortunately,
cache partition strategy would cause load imbalance between the upper-layer cache nodes,
because multiple hot objects can be partitioned to the same upper-layer cache
node, making one cache node become the system bottleneck.

\begin{figure}[t]
\centering
\includegraphics[width=\linewidth]{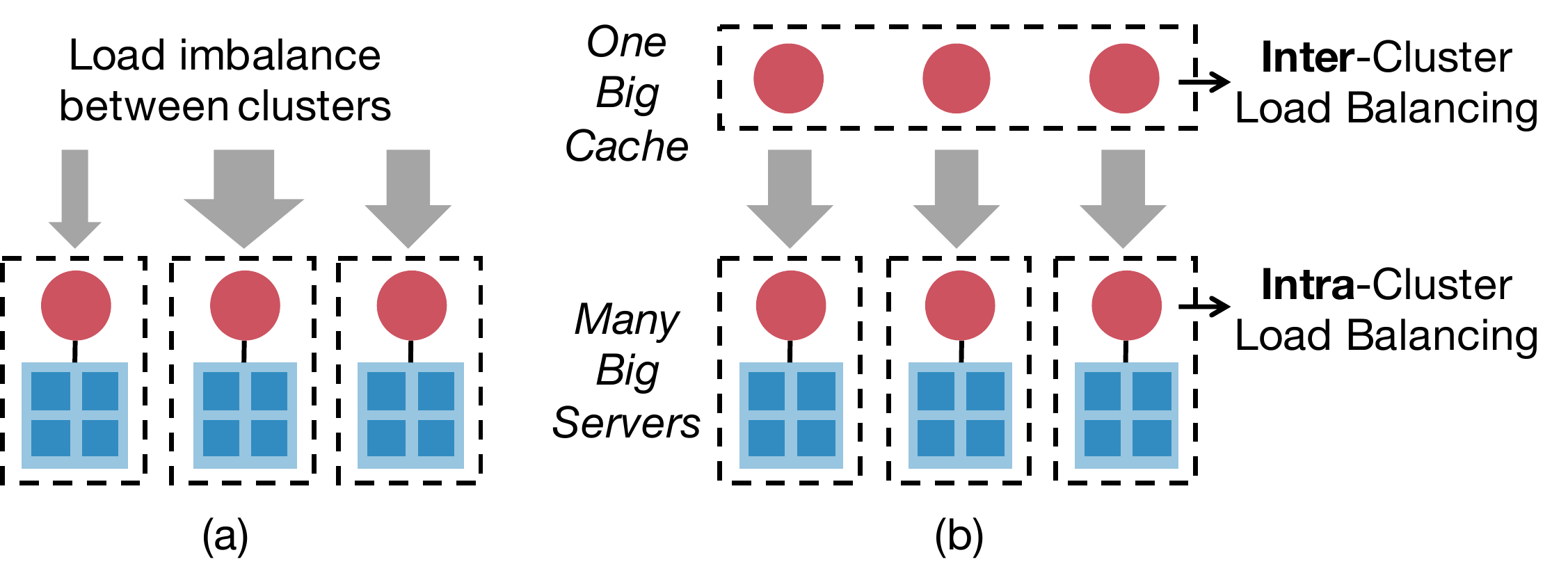}
\vspace{-0.3in}
\caption{Motivation. (a) A cache node only guarantees load balancing for its own
cluster, but the load between clusters can be unbalanced. (b) Use one cache node
in each cluster for \emph{intra-cluster} load balancing, and another layer of cache
nodes for \emph{inter-cluster} load balancing. The challenge is on cache allocation.}
\vspace{-0.2in}
\label{fig:motivation_problem}
\end{figure}

\para{Cache replication.} Cache replication replicates the hot objects to all
the upper-layer cache nodes, and the queries can be uniformly sent to them. As
such, cache replication ensures that the load between the cache nodes is
balanced, and the throughput of caching can grow linearly with the number of cache
nodes. However, cache replication introduces high overhead for cache coherence.
When there is a write query to a cached object, the system needs to update both
the primary copy at the storage node and the cached copies at the cache nodes,
which often requires an expensive two-phase update protocol for cache coherence.
As compared to cache partition which only caches a hot object in one upper-layer
cache node, cache replication needs to update all the upper-layer cache nodes
for cache coherence. This update procedure for cache coherence significantly degrades the throughput of write queries. 

\para{Challenge.} Cache partition has low overhead for cache coherence, but
cannot increase the cache throughput linearly with the number of cache nodes;
cache replication achieves the opposite. Therefore, the main challenge is to
carefully partition and replicate the hot objects, in order to $(i)$ avoid load
imbalance between upper-layer cache nodes, and to $(ii)$ reduce the overhead for
cache coherence.

\section{\sysname Caching Mechanism Design}
\label{sec:mechanism}

\begin{figure}[t]
\centering
\includegraphics[width=\linewidth]{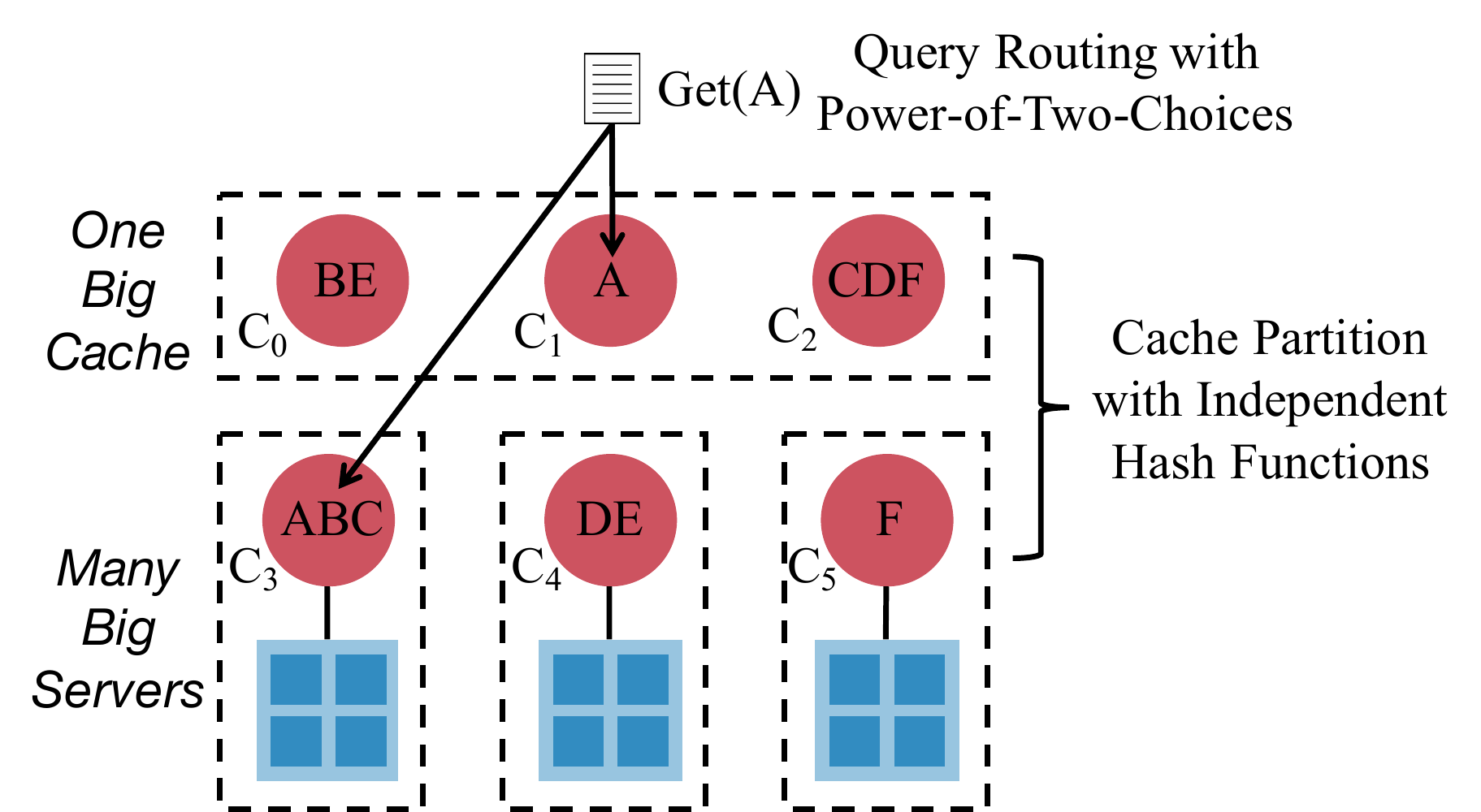}
\vspace{-0.25in}
\caption{Key idea. $(i)$ Use independent hash functions to partition hot objects in
different layers. $(ii)$ Use the power-of-two-choices to route queries,
e.g., route $Get(A)$ to either cache node $C_1$ or cache node $C_3$ based on cache
load.}
\vspace{-0.15in}
\label{fig:design_idea}
\end{figure}

\subsection{Key Idea}
\label{sec:mechanism:idea}

We design \sysname, a new distributed caching mechanism to address the challenge described
in \S\ref{sec:background:scale}. As illustrated by Figure~\ref{fig:design_idea},
our key idea is to use independent hash functions for cache allocation and
the power-of-two-choices for query routing, in order to balance the load
between cache nodes. Our mechanism only caches an object at most once in a
layer, incurring minimal overhead for cache coherence. We first describe the
mechanism and the intuitions, and then show why it works
in \S\ref{sec:mechanism:analysis}.

\para{Cache allocation with independent hash functions.} Our mechanism
partitions the object space with independent hash functions in different layers.
The lower-layer cache nodes primarily guarantee intra-cluster load balancing, each
of which only caches hot objects for its own cluster, and thus each cluster appears
as one ``big server''. The upper-layer cache nodes are primarily for inter-cluster
load balancing, and use a different hash function for partitioning.
The intuition is that if one cache node in a layer is overloaded by
receiving too many queries to its cached objects, because the hash functions of
the two layers are independent, the set of hot objects would be distributed to
multiple cache nodes in another layer with high probability.
Figure~\ref{fig:design_idea} shows an example. While cache node $C_3$ in the lower
layer is overloaded with three hot objects ($A$, $B$ and $C$), the three objects are
distributed to three cache nodes ($C_0$, $C_1$ and $C_2$) in the upper layer.
The upper-layer cache nodes only need to absorb queries for objects (e.g.,  $A$
and $B$) that cause the imbalance \emph{between} the clusters, and do not need to process
queries for objects (e.g., $D$ and $F$) that already \emph{spread out} in the lower-layer
cache nodes.

\para{Query routing with the power-of-two-choices.} The cache allocation
strategy only tells that there exists a way to handle queries without
overloading any cache nodes, but it does not tell how the queries should be
\emph{split} between the layers. Conceivably, we could use a
controller to collect global measurement statistics to infer the query
distribution. Then the controller can compute an optimal solution and enforce it
at the senders. Such an approach has high
system complexity, and the responsiveness to dynamic workloads depends on the
agility of the control loop.

Our mechanism uses an \emph{efficient, distributed, online} solution based on the
power-of-two-choices~\cite{power-of-two} to route queries. Specifically,
the sender of a query only needs to look at the loads of the cache nodes that cache
the queried object, and sends the query to the less-loaded node. For example,
the query $Get(A)$ in Figure~\ref{fig:design_idea} is routed to either $C_1$ or
$C_3$ based on their loads. The key advantage of our solution is that: it is
distributed, so that it does not require a centralized controller or any
coordination between senders; it is online, so that it does not require a
controller to measure the query distribution and compute the solution, and the
senders do not need to know the solution upfront; it is efficient, so that its
performance is close to the optimal solution computed by a controller with
perfect global information (as shown in \S\ref{sec:mechanism:analysis}).
Queries to hit a lower-layer cache node can either pass through an arbitrary
upper-layer node, or totally bypass the upper-layer cache nodes, depending on
the actual use case, which we describe in \S\ref{sec:mechanism:applications}.

\alan{We clarify better idea regarding the cache size.}
\para{Cache size and multi-layer hierarchical caching.} Suppose there are $m$
clusters and each cluster has $l$ servers. First, we let each lower-layer cache
node cache $O(l\log l)$ objects for its own cluster for \emph{intra-cluster}
load balancing, so that a total of $O(ml\log l)$ objects are cached in the
lower layer and each cluster appears like one ``big server''. Then for \emph{inter-cluster}
load balancing, the upper-layer cache nodes only need to cache a total of $O(m\log m)$ objects.
This is different from a single ultra-fast cache at a front-end that handles all $ml$ servers directly.  In that case,
 $O(ml\log (ml))$ objects need to be cached based on the result in~\cite{Fan:smallcache}. However, in \sysname,
 we have an extra upper-layer (with the same total throughput as $ml$ servers) to ``refine'' the query distribution that goes to the lower-layer,
 which reduces the effective cache size in the lower layer to $O(ml\log l)$. Thus, this is not a contradiction with the result in~\cite{Fan:smallcache}.
While these $O(m\log m)$ inter-cluster hot objects also need to
be cached in the lower layer to enable the power-of-two-choices, most of them are also hot inside the
clusters and thus have already been contained in the $O(ml\log l)$ intra-cluster
hot objects. 

Our mechanism can be applied recursively for multi-layer hierarchical caching.
Specifically, applying the mechanism to layer $i$ can balance the load for a set of
``big servers'' in layer $i$-$1$. Query routing
uses the \emph{power-of-k-choices} for \emph{k layers}. Note that using more layers
actually increases the total number of cache nodes, since each layer needs to
provide a total throughput at least equal to that of all storage nodes. The
benefit of doing so is on reducing the cache size. When the number
of clusters is no more than a few hundred, a cache node has enough memory
with two layers.

\subsection{Analysis}
\label{sec:mechanism:analysis}

Prior work~\cite{Fan:smallcache} has shown that caching $O(n\log n)$ hottest objects in a single cache node can balance
the load for $n$ storage nodes for any query
distribution. In our work, we replace the single cache node with multiple cache nodes in two layers to support a larger scale.
Therefore, based on our argument on the cache size in
\S\ref{sec:mechanism:idea}, we need to prove that the two-layer cache can
absorb all queries to the hottest $O(m\log m)$ objects under any query
distribution for all $m$ clusters. We first define a mathematical model to formalize
this problem.

\para{System model.} There are $k$ hot objects $\{o_0, o_1, \dots, o_{k-1}\}$ with query distribution
$P = \{p_0, p_1, \dots, p_{k-1}\}$, where $p_i$ denotes the fraction of queries for
object $o_i$, and $\sum_{i=0}^{k-1} p_i = 1$. The total query rate is $R$, and
the query rate for object $o_i$ is $r_i = p_i \cdot R$.
There are in total $2m$ cache nodes that are organized to two groups $A = \{a_0,
a_1, ..., a_{m-1}\}$ and $B = \{b_0, b_1, ..., b_{m-1}\}$, which represent the upper
and lower layers, respectively. The throughput of each
cache node is $\widetilde{T}$.

The objects are mapped to the cache nodes with two independent hash functions
$h_0(x)$ and $h_1(x)$. Object $o_i$ is cached in $a_{j_0}$ in group $A$
and $b_{j_1}$ in group $B$, where $j_0 = h_0(i)$ and $j_1 = h_1(i)$. A query
to $o_i$ can be served by either $a_{j_0}$ or $b_{j_1}$.

\para{Goal.} Our goal is to evaluate the total query rate $R$ the cache nodes
can support, in terms of $m$ and $\widetilde{T}$, regardless of query
distribution $P$, as well as the relationship between $k$ and $m$. Ideally, we
would like $R \approx
\alpha m \widetilde{T}$ where $\alpha$ is a small constant
(e.g., $1$), so that the operator can easily provision the cache nodes to meet
the cache throughput requirement (i.e., no smaller than the total throughput of
storage nodes).

If we can set $k$ to be $O(m\log m)$, it means that the cache nodes
can absorb all queries to the hottest $O(m\log m)$ objects, despite query
distribution. Combining this result with the cache size argument in \S\ref{sec:mechanism:idea},
we can prove that the distributed caching mechanism can provide
performance guarantees for large-scale storage systems across multiple clusters.

\alan{We clarify our problem can be a matching problem.}
\para{A perfect matching problem in a bipartite graph.} The key observation of
our analysis is that the problem can be converted to finding a perfect matching
in a bipartite graph. Intuitively, if a perfect matching exists, the requests to $k$ hot objects can be completely absorbed from the two layers of cache nodes.
Specifically, we construct a bipartite graph $G = (U, V, E)$,
where $U$ is the set of vertices on the left, $V$ is the set of vertices on the right,
and $E$ is the set of edges. Let $U$ represent the set of objects, i.e.,
$U = \{o_0, o_1, ..., o_{k-1}\}$.
Let $V$ represent the set of cache nodes, i.e.,
$V = A \cup B = \{a_0, a_1, ..., a_{m-1}, b_0, b_1, ..., b_{m-1}\}$.
Let $E$ represent the hash functions mapping from the objects to the cache nodes, i.e.,
$E = \{e_{o_i, a_{j_0}} | h_0(i) = j_0\} \cup \{e_{o_i, b_{j_1}} | h_1(i) = j_1\}$.
Given a query distribution $P$ and a total
query rate $R$, we define a perfect matching in $G$ to represent that the
workload can be supported by the cache nodes.

\begin{definition}
Let $\Gamma(v)$ be the set of neighbors of vertex $v$ in $G$.
A weight assignment $W = \{w_{i,j} \in [0, \widetilde{T}] |
e_{i,j} \in E\}$
is a perfect
matching of $G$ if
\begin{enumerate}
    \item $\forall o_i \in U: \quad \sum_{v \in \Gamma(o_i)} w_{o_i,v} = p_i \cdot R$, and
    \item $\forall v \in V: \quad \sum_{u \in \Gamma(v)} w_{u,v} \leq \widetilde{T}$.
\end{enumerate}
\end{definition}

In this definition, $w_{i,j}$ denotes the portion of the queries to object $i$ served
by cache node $j$. Condition 1 ensures that for any object $o_i$, its query
rate $p_i \cdot R$ is fully served. Condition 2 ensures that for any cache
node $v$, its load is no more than $\widetilde{T}$, i.e., no single cache node is overloaded.

\begin{figure}[t]
\centering
\includegraphics[width=\linewidth]{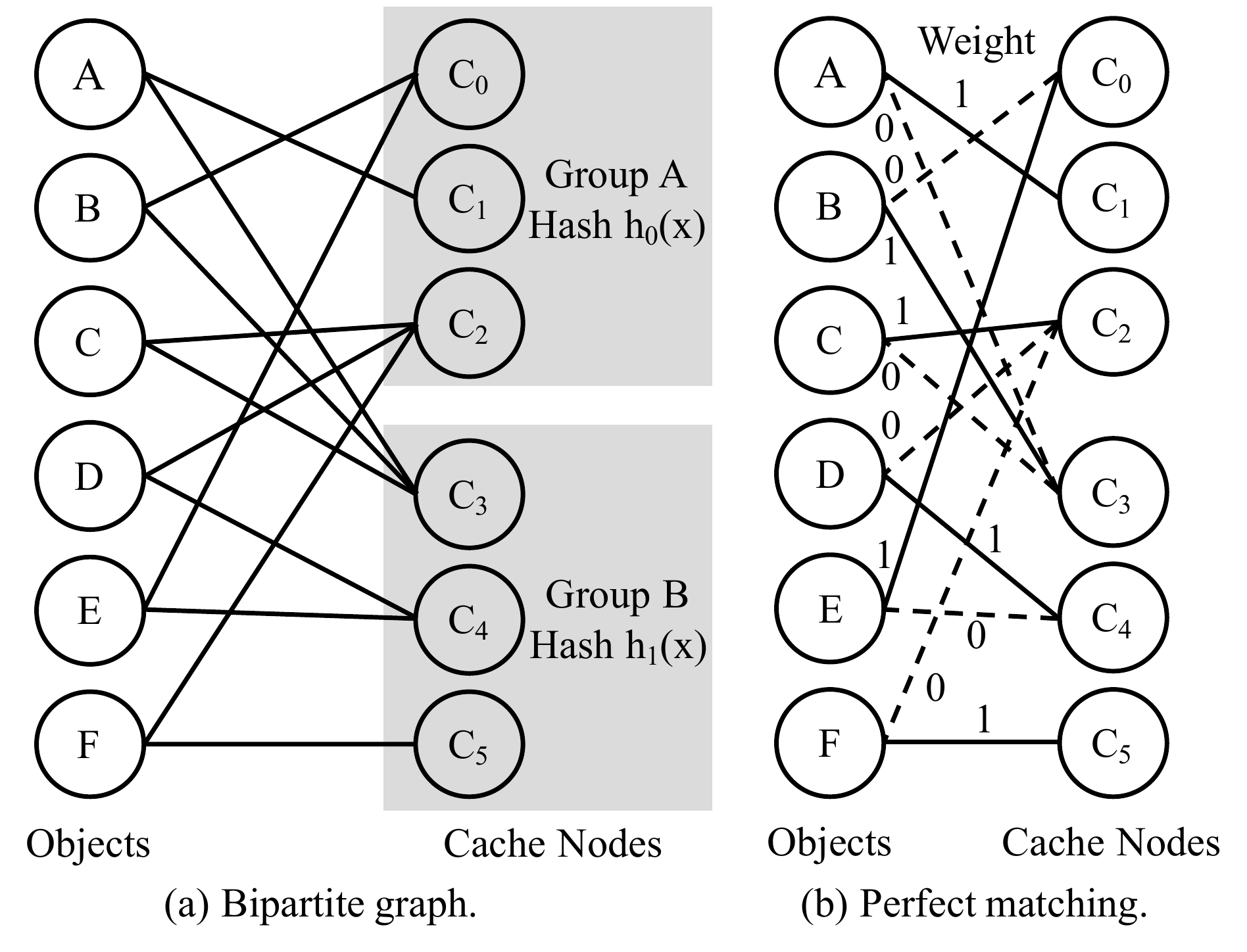}
\vspace{-0.34in}
\caption{Example for analysis. (a) A bipartite graph constructed for the scenario
in Figure~\ref{fig:design_idea}. (b) A perfect matching for query routing when all
objects have a query rate of 1, and all cache nodes have a throughput of 1.}
\vspace{-0.1in}
\label{fig:analysis_example}
\end{figure}

\emph{When a perfect matching exists, it is feasible to serve all the queries by
the cache nodes.} We use the example in Figure~\ref{fig:analysis_example} to
illustrate this. Figure~\ref{fig:analysis_example}(a) shows the bipartite graph
constructed for the scenario in Figure~\ref{fig:design_idea}, which contains six
hot objects ($A$-$F$) and six cache nodes
in two layers ($C_0$-$C_5$). The edges are built based on two hash functions
$h_0(x)$ and $h_1(x)$. Figure~\ref{fig:analysis_example}(b) shows a perfect
matching for the case that all objects have the same
query rate $r_i = 1$ and all cache nodes have the same throughput $\widetilde{T}
= 1$. The number besides an edge denotes the weight of an edge, i.e., the rate of
the object served by the cache node. For instance, all queries to $A$ are served
by $C_1$. This is a simple example to
illustrate the problem. In general, the query rates of the objects do not have
to be the same, and the queries to one object may be served by multiple cache
nodes.


\para{Step 1: existence of a perfect matching.} We first show the existence of a perfect matching for any given total rate $R$ and
\emph{any} query distribution $P$. We have the
following lemma to demonstrate how big the total rate $R$ can be in terms of
$\widetilde{T}$, for any $P$. For the full proof of Lemma~\ref{lem:feasible},
we refer the readers to \S{A.2} in the Appendix. \alan{We will upload report to arxiv and update the citation}

\begin{lemma}\label{lem:feasible}
Let $\alpha$ be a suitably small constant. If $k \leq m^\beta$ for some constant
$\beta$ (i.e., $k$ and $m$ are polynomial-related)
and $\max_i(p_i) \cdot R \leq \widetilde{T}/2$,
then for any $\epsilon > 0$, there exists a perfect matching
for $R = (1-\epsilon)\alpha \cdot m\widetilde{T}$ and any $P$, with high probability
for sufficiently
large $m$.
\end{lemma}
\alan{We give proof sketch of Lemma 1.}
\emph{Proof sketch of Lemma~\ref{lem:feasible}:}
We utilize the results and techniques developed from expander graphs and network
flows. $(i)$ We first show that $G$ has the
so-called \emph{expansion property} with high probability. Intuitively, the
property states that the neighborhood of any subset of nodes in $U$
\emph{expands}, i.e., for any $S \subseteq U$, $|\Gamma(S)|\geq |S|$.
It has been observed that such properties exist in a wide
range of random graphs~\cite{vadhan2012pseudorandomness}. While our $G$ behaves similar to random
bipartite graphs, we need the expansion property to hold for $S$ in any size,
which is stricter than the standard definition (which assumes $S$ is not too
large) and thus requires more delicate probabilistic techniques.
$(ii)$ We then show that if a graph has the expansion property, then it has a perfect matching.
This step can be viewed as a generalization of Hall's theorem~\cite{nevsetril2011graph}
in our setting. Hall's theorem states that a balanced bipartite graph has
a perfect (non-fractional) matching if and only if for any subset $S$ of the
left nodes, $|\Gamma(S)| \geq |S|$. 
and perfect matching can be fractional. This step can be proved by the
max-flow-min-cut theorem, i.e., expansion implies large cut, and then implies
large matching.

\para{Step 2: finding a perfect matching.} Demonstrating the existence of a perfect
matching is insufficient since it just ensures the queries can be absorbed but does not give the actual weight assignment $W$, i.e., how
the cache nodes should serve queries for each $P$ to achieve $R$. This means
that the system would require an algorithm to compute $W$ and an mechanism to
enforce $W$. As discussed in \S\ref{sec:mechanism:idea}, instead of doing so, we
use the power-of-two-choices to ``emulate'' the perfect matching, without the
need to know what the perfect matching is. The quality of the mechanism is backed
by Lemma~\ref{lem:stationary}, which we prove using queuing theory. The detailed proof
can be found in \S{A.3}. \alan{We will upload report to arxiv and update the citation}


\begin{lemma}\label{lem:stationary}
If a perfect matching exists for $G$, then the power-of-two-choices process is stationary.
\label{lem:pot}
\end{lemma}

Stationary means that the load on the cache nodes would converge, and the system
is ``sustainable'' in the sense that the system will never ``blow up'' (i.e.,
build up queues in a cache node and eventually drop queries) with query rate
$R$.

\alan{We give proof sketch of Lemma 2.}
\emph{Proof sketch of Lemma~\ref{lem:stationary}:} Showing this lemma requires us to use a powerful building block in query theory presented in
~\cite{foss1998stability,foley2001join}. Consider $2m$ exponential random variables with rate $\widetilde T_i > 0$. Each non-empty set of cache nodes $S \subseteq [2m]$, has an associated Poisson arrival process with rate $\lambda_S \geq 0$ that joins the shortest queue in $S$ with ties broken randomly. For each non-empty subset $Q \subseteq [2m]$, define the traffic intensity on $Q$ as

\begin{equation*}\label{eqn:intensity_paper}
\rho_Q = \frac{\sum_{S \subseteq Q}\lambda_S}{\mu_Q},
\end{equation*}

where $\mu_Q = \sum_{i \in Q}\widetilde T_i$. Note that the total rate at which  objects  served by $Q$ can be greater than the numerator of (\ref{eqn:intensity_paper}) since other requests may be allowed to be served by some or all of the cache nodes in $Q$. Let $\rho_{\max} = \max_{Q\subseteq[2m]}\{\rho_Q\}$.  Given the result in~\cite{foss1998stability,foley2001join},  if we can show $\rho_{\max} < 1$, then the Markov process is positive recurrent and has a stationary distribution. In fact, our cache querying can be described as the following arrival process.

Define $D(i) = \{a_{h_0(i)}, b_{h_1(i)}\}$. Let $S$ be an arbitrary subset of $\{A, B\}$. Define $\lambda_S$ as:
\begin{itemize}
\item If $S = \{a_i, b_j\}$ for some $i$ and $j$, let
\begin{equation*}
\lambda_S = \sum_{i \leq k}(I(D(i) = S)) p_i R,
\end{equation*}
where $I(\cdot)$ is an indicator function that sets to $1$ if and only if its argument is true.
\item Otherwise, $\lambda_S = 0$.
\end{itemize}
Finally, we show that $\rho_{\max}$ is less than 1 (refer to \S{A.2}) and thus the process is stationary.

\para{Step 3: main theorem.} Based on Lemma~\ref{lem:feasible} and Lemma~\ref{lem:pot}, we
can prove that our distributed caching mechanism is able to provide a performance
guarantee, despite query distribution.

\begin{theorem}\label{thm:main}
Let $\alpha$ be a suitable constant.
If $k \leq m^\beta$ for some constant $\beta$ (i.e., $k$ and $m$ are polynomial-related)
and $\max_i(p_i) \cdot R \leq \widetilde{T}/2$, then for any $\epsilon > 0$,
the system is stationary for $R = (1-\epsilon)\alpha \cdot m\widetilde{T}$ and any $P$,
with high probability for sufficiently large $m$.
\label{the:balance}
\end{theorem}

\alan{We add a quick interpretation of our analysis.}
\emph{Interpretation of the main theorem:} As long as the query rate of a single hot object  $o_i$ is no larger than $ \widetilde{T}/2$ (e.g., half of the entire throughput in a cluster rack), \sysname can support a query rate of $\approx m\widetilde{T}$ for any query distributions to the $k$ hot objects (where $k$ can be fairly large in terms of $m$) by using the power-of-two-choices protocol to route the queries to the cached objects. The key takeaways are presented in the following section. 

\subsection{Remarks}
\label{sec:mechanism:remarks}

Our problem isn't a balls-in-bins problem using the original power-of-two-choices.
The major difference is that our problem hashes objects into cache
nodes, and queries to the same object by \emph{reusing} the same hash
functions, instead of using a \emph{new random
source} to sample two nodes for each query. In fact, without using the power-of-two-choices,
the system is in \emph{non-stationary}. This means that the power-of-two-choices
makes a ``life-or-death'' improvement in our problem, instead of a ``shaving off a
log $n$'' improvement. While we refer to the Appendix for detailed discussions, we have a few important remarks. \begin{itemize}
[leftmargin=*]
    \item \textbf{Nonuniform number of cache nodes in two layers.} For simplicity
    we use the same number of $m$ cache nodes per layer in the system. However, we can generalize the
    analysis to accommodate the cases of different numbers of caches nodes in two layers, as long as $\min (m_0,
    m_1)$ is sufficiently large, where $m_0$ and $m_1$ are the number of
    upper-layer and lower-layer cache nodes respectively. While it requires $m$
    to be sufficiently large, it is not a severe restriction, because the load
    imbalance issue is only significant when $m$ is large.

    \item \textbf{Nonuniform throughput of cache nodes in two groups.} Although our analysis assumes
    the throughput of a cache node is $\widetilde{T}$,
    we can generalize it to accommodate the cases of nonuniform throughput
    by treating a cache node with a large throughput as multiple
    smaller cache nodes with a small throughput.

    \item \textbf{Cache size.} As long as the number of objects and the number
    of cache nodes are polynomially-related ($k \leq m^\beta$), the system is able
    to provide the performance guarantee. It is more relaxed than $O(m\log m)$.
    Therefore, by setting $k = O(m\log m)$, the cache nodes are able to absorb
    all queries to the hottest $O(m\log m)$ objects, making the load on the $m$
    clusters balanced.

    \item \textbf{Maximum query rate for one object.} The theorem requires that
    the maximum query rate for one object is no bigger than half the throughput of
    one cache node.
    This is not a strict requirement for the system, because a cache
    node is orders of magnitude faster than a storage node.

    \item \textbf{Performance guarantee.} The system can guarantee a total throughput
    of $R = (1-\epsilon)\alpha \cdot mT$, which scales linearly with $m$ and $T$.
    In practice, $\alpha$ is close to $1$.
\end{itemize}

\subsection{Use Cases}
\label{sec:mechanism:applications}

\sysname is a general solution that can be applied to scale out various storage
systems (e.g., key-value stores and file systems) using different storage
mediums (e.g., HDD, SDD and DRAM). We describe
two use cases.

\para{Distributed in-memory caching.} Based on the performance gap between DRAMs and
SSDs, a fast in-memory cache node can be used balance an SSD-based storage cluster,
such as SwitchKV~\cite{switchkv}. \sysname can scale out SwitchKV by using another
layer of in-memory cache nodes to balance multiple SwitchKV clusters. While it is
true that multiple in-memory cache nodes can be balanced using a faster switch-based
cache node, applying \sysname obviates the need to introduce a new component (i.e.,
a switch-based cache) to the system. Since the queries are routed to the
cache and storage nodes by the network, queries to the lower-layer cache nodes
can totally bypass the upper-layer cache nodes.

\para{Distributed switch-based caching.} Many low-latency storage systems for interactive
web services use more expensive in-memory designs. An in-memory storage rack can
be balanced by a switch-based cache like NetCache~\cite{netcache}, which directly
caches the hot objects in the data plane of the ToR switch. \sysname can scale out
NetCache to multiple racks by caching hot objects in a higher layer of the network topology,
e.g., the spine layer in a two-layer leaf-spine network. As discussed in the
remarks (\S\ref{sec:mechanism:remarks}), \sysname accommodates the cases that the number
of spine switches is smaller and each spine switch is faster. As for query routing,
while queries to hit the leaf cache switches need to inevitably go through the
spine switches, these queries can be arbitrarily routed through any spine
switches, so that the load on the spine switches can be balanced.

\begin{figure}[t]
\centering
\includegraphics[width=0.9\linewidth]{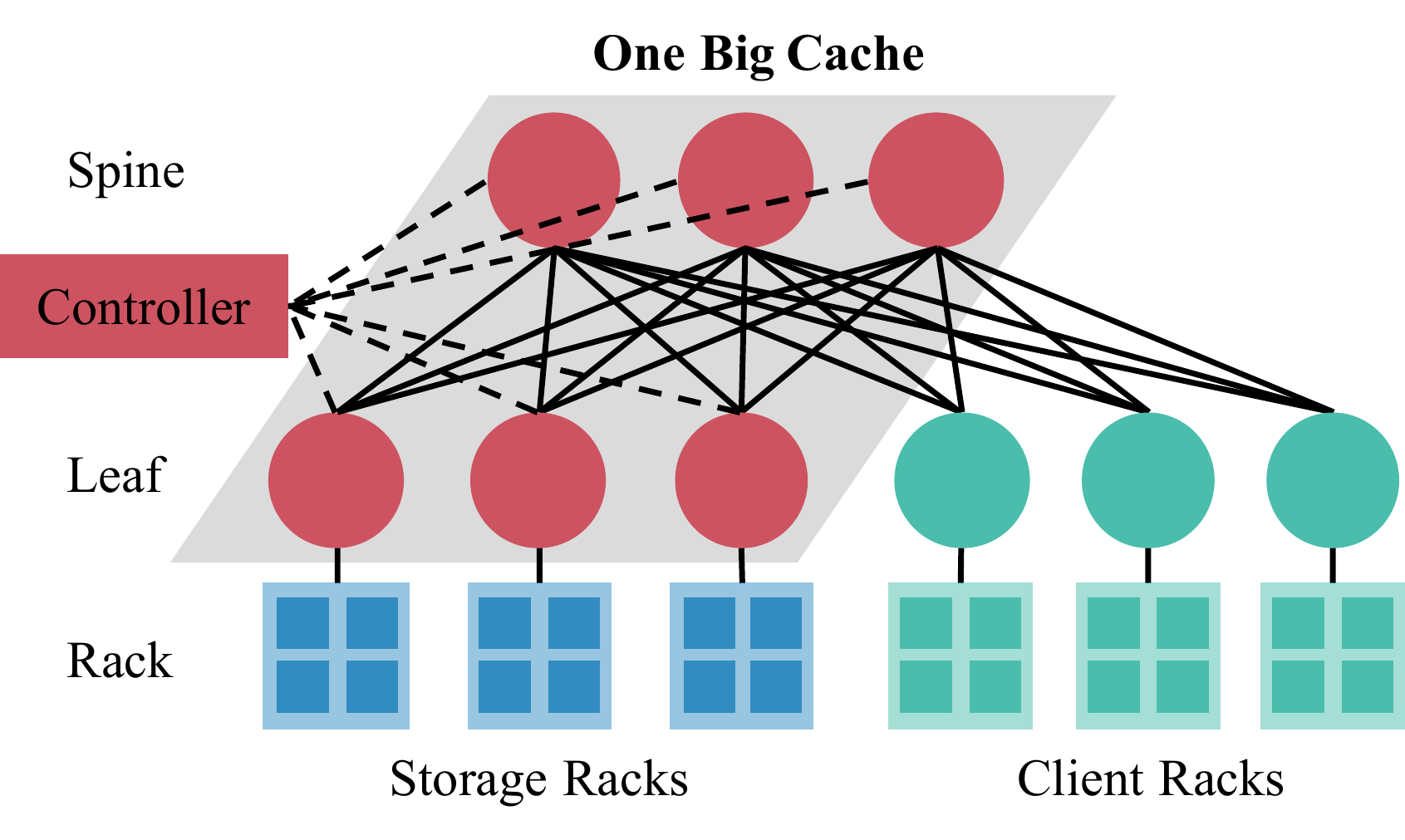}
\vspace{-0.20in}
\caption{Architecture for distributed switch-based caching.}
\vspace{-0.15in}
\label{fig:architecture}
\end{figure}

Note that while existing solutions (e.g., NetCache~\cite{netcache}) directly embeds
caching in the switches which may raise concerns on deployment, another option
for easier deployment is to use the cache switches as \emph{stand-alone specialized
appliances} that are separated from the switches in the datacenter network.
\sysname can be applied to scale out these specialized switch-based caching
appliances as well.

\section{\sysname for Switch-Based Caching}
\label{sec:design}

To demonstrate the benefits of \sysname, we provide a concrete system design
for the emerging switch-based caching. A similar design can be applied
to other use cases as well.
\vspace{-0.05in}


\begin{figure*}[t]
\centering
\includegraphics[width=0.95\linewidth]{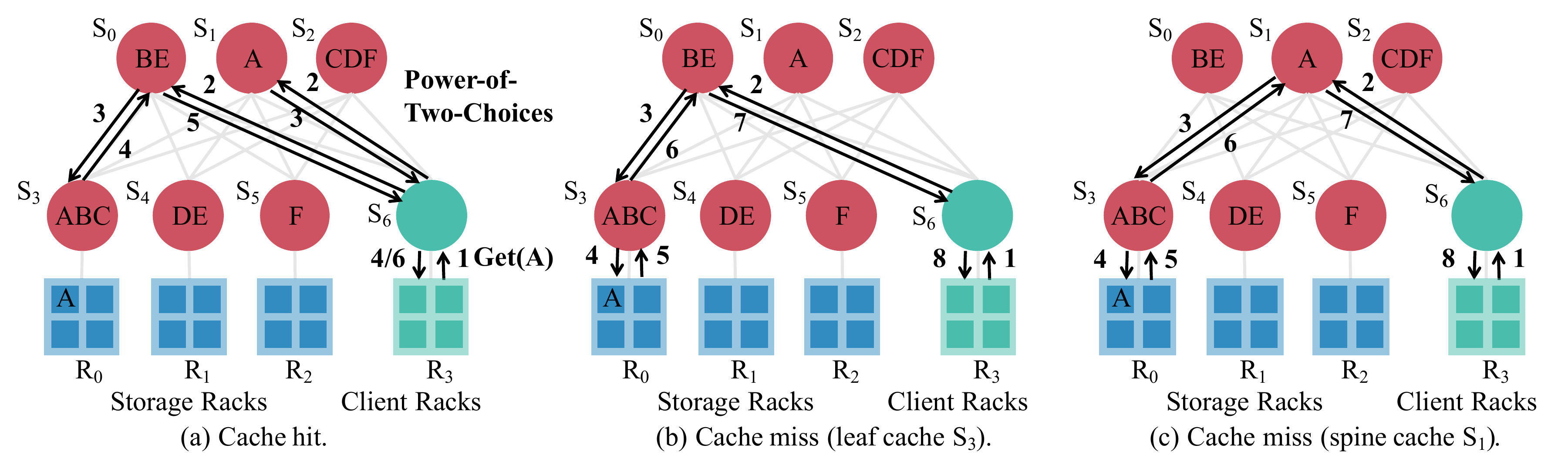}
\vspace{-0.2in}
\caption{Query handling for $Get(A)$. $S_6$ uses the power-of-two-choices to
decide whether to send $Get(A)$ to $S_1$ or $S_3$. (a) Upon a cache hit, the
switch directly replies the query, without visiting the storage server. (b, c)
Upon a cache miss, the query is forwarded to the storage server without routing
detour.}
\vspace{-0.1in}
\label{fig:query}
\end{figure*}

\subsection{System Architecture}
\label{sec:design:architecture}

Emerging switch-based caching, such as NetCache~\cite{netcache} is limited to one storage
rack. We apply \sysname to switch-based caching to provide load balancing for cloud-scale
key-value stores that span many racks. Figure~\ref{fig:architecture}
shows the architecture for a two-layer leaf-spine datacenter network.

\para{Cache Controller.} The controller computes the cache partitions, and notifies the
cache switches. It updates the cache allocation
under system reconfigurations, e.g., adding new racks and cache switches, and
system failures; and thus updating the allocation is an infrequent task. We assume the controller is reliable by replicating on
multiple servers with a consensus protocol such as Paxos~\cite{paxos98}.
The controller is not involved in handling storage queries in the data plane.

\para{Cache switches.} The cache switches provide two critical functionalities
for \sysname: (1) caching hot key-value objects; (2) distributing switch load information
for query routing. First, a local agent in the switch OS receives its cache
partition from the controller, and manages the hot objects for its partition in
the data plane. Second, the cache switches
implement a lightweight in-network telemetry mechanism to distribute their load
information by piggybacking in packet headers. The functionalities for
\sysname are invoked by a reserved L4 port, so that \sysname does not affect other
network functionalities. We use existing L2/L3 network protocols to route
packets, and do not modify other network functionalities already in the switch.

\para{ToR switches at client racks.} The ToR switches at client racks provide query
routing. It uses the power-of-two-choices to
decide which cache switch to send a query to, and uses existing L2/L3 network
protocols to route the query.

\para{Storage servers.} The storage servers host the key-value store. \sysname
runs a shim layer in each storage server to integrate the in-network cache with
existing key-value store software like Redis~\cite{redis} and
Memcached~\cite{memcached}. The shim layer also implements a cache coherence
protocol to guarantee the consistency between the servers and cache switches.

\para{Clients.} \sysname provides a client library for applications to access the
key-value store. The library provides an interface similar to existing key-value
stores. It maps function calls from applications to \sysname
query packets, and gathers \sysname reply packets to generate function returns.

\subsection{Query Handling}
\label{sec:design:query}

A key advantage of \sysname is that it provides a distributed \emph{on-path}
cache to serve queries at line rate. Read queries on cached objects (i.e.,
cache hit) are directly replied by the cache switches, \emph{without the need
to visit storage servers}; read queries on uncached objects (i.e., cache miss)
and write queries are forwarded to storage servers, \emph{without any routing
detour}. Further, while the cache is distributed, our query routing mechanism
based on the power-of-two-choices ensures that the load between the cache
switches is balanced.

\para{Query routing at client ToR switches.} Clients send queries via the client
library, which simply translates function calls to query packets. The complexity
of query routing is done at the ToR switches of the client racks. The ToR
switches use the switch on-chip memory to store the loads of the cache switches.
For each read query, they compare the loads of the switches that contain the
queried object in their partitions, and pick the less-loaded cache switch for
the query. After the cache switch is chosen, they use the existing routing
mechanism to send the query to the cache switch. The routing mechanism can pick
a routing path that balances the traffic in the network, which is orthogonal to
this paper. Our prototype uses a mechanism similar to CONGA~\cite{conga} and
HULA~\cite{hula} to choose the least loaded path to the cache switch.

For a cache hit, the cache switch copies the value from its on-chip memory to
the packet, and returns the packet to the client. For a cache miss, the cache
switch forwards the packet to the corresponding storage server that stores the
queried object. Then the server processes the query and replies to the client.
Figure~\ref{fig:query} shows an example. A client in rack $R_3$ sends a query to
read object $A$. Suppose $A$ is cached in switch $S_1$ and $S_3$, and is stored
in a server in rack $R_0$. The ToR switch $S_6$ uses the power-of-two-choices to
decide whether to choose $S_1$ or $S_3$. Upon a cache hit, the cache switch
(either $S_1$ or $S_3$) directly replies to the client
(Figure~\ref{fig:query}(a)). Upon a cache miss, the query is sent to the server.
But no matter whether the leaf cache (Figure~ \ref{fig:query}(b)) or the spine
cache (Figure~\ref{fig:query}(c)) is chosen, there is no routing detour for the
query to reach $R_0$ after a cache miss.

Write queries are directly forwarded to the storage servers that
contain the objects. The servers implement a cache coherence protocol for data consistency as
described in \S\ref{sec:design:cache}.

\para{Query processing at cache switches.} Cache switches use the on-chip memory
to cache objects in their own partitions. In programmable switches such as
Barefoot Tofino~\cite{tofino}, the on-chip memory is organized as register
arrays spanning multiple stages in the packet processing pipeline. The packets
can read and update the register arrays at line rate. We uses the same
mechanism as NetCache~\cite{netcache} to implement a key-value cache that can
support variable-length values, and a heavy-hitter (HH) detector that the switch
local agent uses to decide what top $k$ hottest objects in its partition to cache. 

\para{In-network telemetry for cache load distribution.} We use a
light-weight in-network telemetry mechanism to distribute the cache load
information for query routing. The mechanism piggybacks the switch load (i.e.,
the total number of packets in the last second) in the packet headers of reply
packets, and thus
incurs minimal overhead. Specifically, when a reply packet of a query passes a
cache switch, the cache switch adds its load to the packet header. Then when the
reply packet reaches the ToR switch of the client rack, the ToR switch retrieves
the load in the packet header to update the load stored in its on-chip memory.
To handle the case that the cache load may become stale without enough traffic
for piggybacking, we can add a simple aging mechanism that would gradually
decrease a load to zero if the load is not updated for a long time. Note that
aging is commonly supported by modern switch ASICs, but it is not supported by
P4 yet, and thus is not implemented in our prototype.

\subsection{Cache Coherence and Cache Update}
\label{sec:design:cache}

\para{Cache coherence.} Cache coherence ensures data consistency between storage
servers and cache switches when write queries update the values of the objects. The
challenge is that an object may be cached in multiple cache switches, and need
to be updated atomically. Directly updating the copies of an object in the cache
switches may result in data inconsistency. This is because the cache switches
are updated asynchronously, and during the update process, there would be a mix
of old and new values at different switches, causing read queries to get
different values from different switches.

\begin{figure*}[t]
\centering
\includegraphics[width=0.95\linewidth]{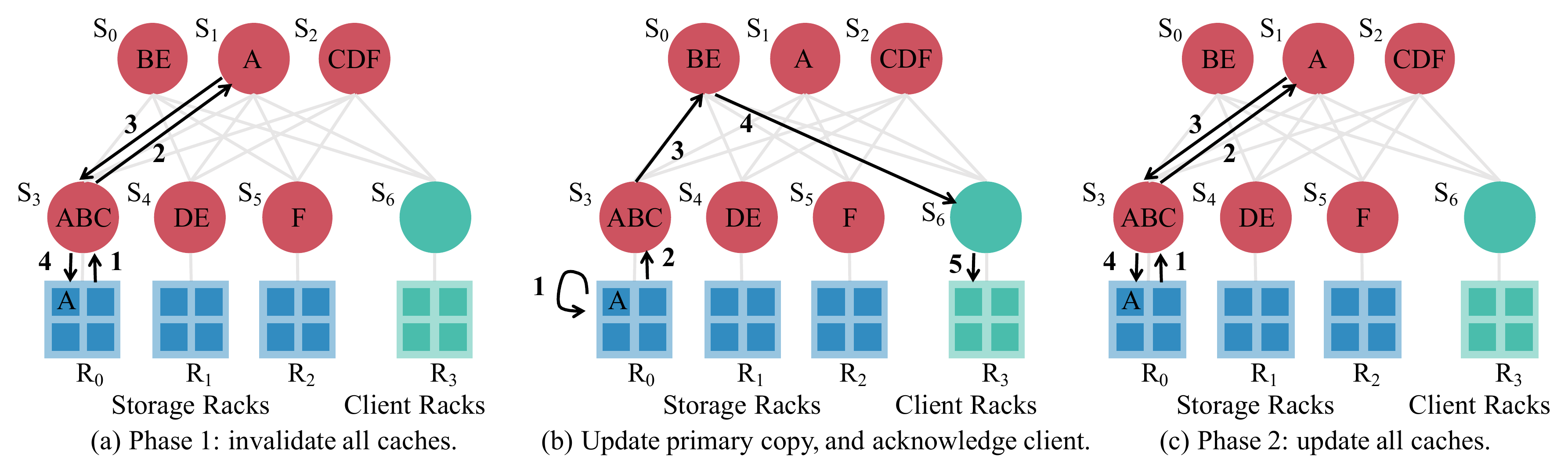}
\vspace{-0.2in}
\caption{Cache coherence is achieved by a two-phase update protocol in \sysname.
The example shows the process to handle an update to object $A$ stored in rack $R_0$
with the two-phase update protocol.}
\vspace{-0.1in}
\label{fig:cache_coherence}
\end{figure*}

We leverage the classic two-phase update protocol~\cite{Bernstein:1986} to ensure strong consistency,
where the first phase invalidates all copies and the second phase updates all
copies. To apply the protocol to our scenario, after receiving a write query,
the storage server generates a packet to invalidate the copies in the cache
switches. The packet traverses a path that includes all the switches that cache the
object. The return of the invalidation packet
indicates that all the copies are invalidated. Otherwise, the server resends
the invalidation packet after a timeout. Figure~\ref{fig:cache_coherence}(a)
shows an example that the copies of object $A$ are invalidated by an
invalidation packet via path $R_0$-$S_3$-$S_1$-$S_3$-$R_0$. After the first
phase, the server can update its primary copy, and send an acknowledgment to the
client, instead of waiting for the second phase, as illustrated by
Figure~\ref{fig:cache_coherence}(b). This optimization is safe, since all copies
are invalid. Finally, in the second phase, the server sends an update packet to
update the values in the cache switches, as illustrated by
Figure~\ref{fig:cache_coherence}(c).

\para{Cache update.} The cache update is performed in a decentralized way
without the involvement of the controller. We use a similar mechanism
as NetCache~\cite{netcache}. Specifically, the local agent in each switch uses
the HH detector in the data plane to detect hot objects in its own partition,
and decides cache insertions and evictions.
Cache evictions can be directly done by the agent;
cache insertions require the agent to contact the storage
servers. Slightly different from NetCache,
\sysname uses a cleaner, more efficient mechanism to unify cache insertions and
cache coherence. Specifically, the agent first inserts the new object into the
cache, but marks it as invalid. Then the agent notifies the server; the server
updates the cached object in the data plane using phase 2 of cache coherence,
and serializes this operation with other write queries. As for comparison, in
NetCache, the agent copies the value from the server to the switch via the
switch control plane (which is slower than the data plane), and during the
copying, the write queries to the object are blocked on the server.

\subsection{Failure Handling}
\label{sec:design:failure}

\paraf{Controller failure.} The controller is replicated on multiple servers for
reliability (\S\ref{sec:design:architecture}). Since the controller is only
responsible for cache allocation, even if all servers of the controller fail, the
data plane is still operational and hence processes queries. The servers can be
simply rebooted.

\para{Link failure.} A link failure is handled by existing network protocols,
and does not affect the system, as long as the network is connected and the routing
is updated. If the network is partitioned after a link failure,
the operator would choose between consistency and availability, as stated by the CAP
theorem. If consistency were chosen, all writes should be blocked; if
availability were chosen, queries can still be processed, but cache coherence
cannot be guaranteed.

\para{ToR switch failure.} The servers in the rack would lose access to the network.
The switch needs to be rebooted or replaced. If the switch is in a storage rack,
the new switch starts with an empty cache and uses the cache update process
to populate its cache. If the switch is in a client rack, the new switch initializes
the loads of all cache switches to be zero, and uses the in-network telemetry mechanism
to update them with reply packets.

\para{Other Switch failure.} If the switch is not a cache switch, the failure is
directly handled by existing network protocols. If the switch is a cache switch,
the system loses throughput provided by this switch. If it can be quickly
restored (e.g., by rebooting), the system simply waits for the switch to come
back online. Otherwise, the system remaps the cache partition of the failed
switch to other switches, so that the hot objects in the failed switch can still
be cached, alleviating the impact on the system throughput. The remapping
leverages consistent hashing~\cite{consistent-hashing} and virtual
nodes~\cite{Dabek:wcs} to spread the load. Finally, if the network is
partitioned due to a switch failure, the operator would choose consistency or
availability, similar to that of a link failure.

\section{Implementation}
\label{sec:implementation}

We have implemented a prototype of \sysname to realize distributed switch-based
caching, including cache switches, client ToR switches, a controller, storage
servers and clients. 

\para{Cache switch.} The data plane of the cache switches is written in
the P4 language~\cite{p4-ccr}, which is a domain-specific language to program the packet forwarding pipelines of data plane devices. 
P4 can be used to program the switches that are based on Protocol Independent Switch Architecture (PISA). In this architecture, 
we can define the packet formats and packet processing behaviors by a series of match-action tables. 
These tables are allocated to different processing stages in a forwarding pipeline, based on hardware resources. 
Our implementation is compiled to Barefoot Tofino ASIC~\cite{tofino} with
Barefoot P4 Studio software suite~\cite{P4Studio}. In the Barefoot Tofino switch, 
we implement a key-value cache module uses 16-byte keys, and contains 64K 16-byte slots per
stage for 8 stages, providing values at the granularity of 16 bytes and up to 128
bytes without packet recirculation or mirroring. The Heavy Hitter 
detector module contains a Count-Min sketch~\cite{CMSketch}, which has 4 register arrays and
64K 16-bit slots per array, and a Bloom filter, which has 3 register arrays and
256K 1-bit slots per array. The telemetry module uses one 32-bit register slot
to store the switch load. We reset the counters in the HH detector and telemetry modules in every second. The local agent in the switch OS is written in Python. 
It receives cache partitions from the controller, and manages the switch ASIC
via the switch driver using a Thrift API generated by the P4 compiler. The
routing module uses standard L3 routing which forwards packets based on
destination IP address. 

\para{Client ToR switch.} The data plane of client ToR switches is also written
in P4~\cite{p4-ccr} and is compiled to Barefoot Tofino ASIC~\cite{tofino}. Its
query routing module contains a register array with 256 32-bit slots to store
the load of cache switches. The routing module uses standard L3 routing, and
picks the least loaded path similar to CONGA~\cite{conga} and HULA~\cite{hula}.

\para{Controller, storage server, and client. } The controller is written in Python. It computes cache partitions
and notifies the result to switch agents through Thrift API. The shim layer at each storage server
implements the cache coherence protocol, and uses the hiredis
library~\cite{hiredis} to hook up with Redis~\cite{redis}. The client
library provides a simple key-value interface. We use the client library to
generate queries with different distributions and different write ratios.

\section{Evaluation}\label{sec:evaluation}
\subsection{Methodology}
\label{sec:evaluation:methodology}

\paraf{Testbed.} Our testbed consists of two 6.5Tbps Barefoot Tofino switches
and two server machines. Each server machine is equipped with a 16 core-CPU
(Intel Xeon E5-2630), 128 GB total memory (four Samsung 32GB DDR4-2133 memory),
and an Intel XL710 40G NIC.

The goal is to apply \sysname to switch-based caching to provide load balancing for
cloud-scale in-memory key-value stores.
Because of the limited hardware resources we have, we are unable to evaluate
\sysname at full scale with tens of switches and hundreds of servers.
Nevertheless, we make the most of our testbed to evaluate \sysname by dividing
switches and servers into multiple logical partitions and running real switch
data plane and server software, as shown in Figure~\ref{fig:eval_testbed}.
Specifically, a physical switch emulates several virtual switches by using
multiple queues and uses counters to rate limit each queue. We use one Barefoot
Tofino switch to emulate the spine switches, and the other to emulate the leaf
switches. Similarly, a physical server emulates several virtual servers by using
multiple queues. We use one server to
emulate the storage servers, and the other to emulate the clients. We would like
to emphasize that the testbed runs the real switch data plane and runs the Redis key-value
store~\cite{redis} to process real key-value queries.

\para{Performance metric.} By using multiple processes and using the pipelining
feature of Redis, our Redis server can achieve a throughput of 1 MQPS.
We use Redis to demonstrate that \sysname can integrate with production-quality open-source
software that is widely deployed in real-world systems.
We allocate the 1 MQPS throughput to the emulated storage servers equally with
rate limiting. Since a switch is able to process a few BQPS, the
bottleneck of the testbed is on the Redis servers. Therefore, we use rate
limiting to match the throughput of each emulated switch to the aggregated
throughput of the emulated storage servers in a rack. We normalize the system throughput
to the throughput of one emulated key-value server as the performance metric.

\para{Workloads.} We use both uniform and skewed workloads in the evaluation. The
uniform workload generates queries to each object with the same probability. The skewed
workload follows Zipf distribution with a skewness parameter (e.g., 0.9, 0.95, 0.99).
Such skewed workload is commonly used to benchmark key-value
stores~\cite{switchkv, eccache}, and is backed by measurements from production
systems~\cite{workload-fb-sigmetrics12, benchmarking-socc10}. The clients use approximation
techniques~\cite{switchkv, Gray:1994} to quickly generate queries according
to a Zipf distribution. We store a total of 100 million objects in the key-value store.
We use Zipf-0.99 as the default query distribution to show that \sysname performs well
even under extreme scenarios. We vary the skewness and the write ratio (i.e., the
percentage of write queries) in the experiments to evaluate the performance of \sysname
under different scenarios.

\para{Comparison.} \alan{explain the compared solutions quickly.}To demonstrate the benefits of \sysname, we compare
the following mechanisms in the experiments: \sysname, \creplication,
\cpartition, and \nocache. As described in \S\ref{sec:background:scale}, \creplication is to replicate the hot objects to 
all the upper layer cache nodes, and \cpartition partitions the hot objects between nodes. In \nocache, we do not cache any
objects in both layers. Note that \cpartition performs the same as only using NetCache for each
rack (i.e., only caching in the ToR switches).

\begin{figure}[t]
\centering
\includegraphics[width=\linewidth]{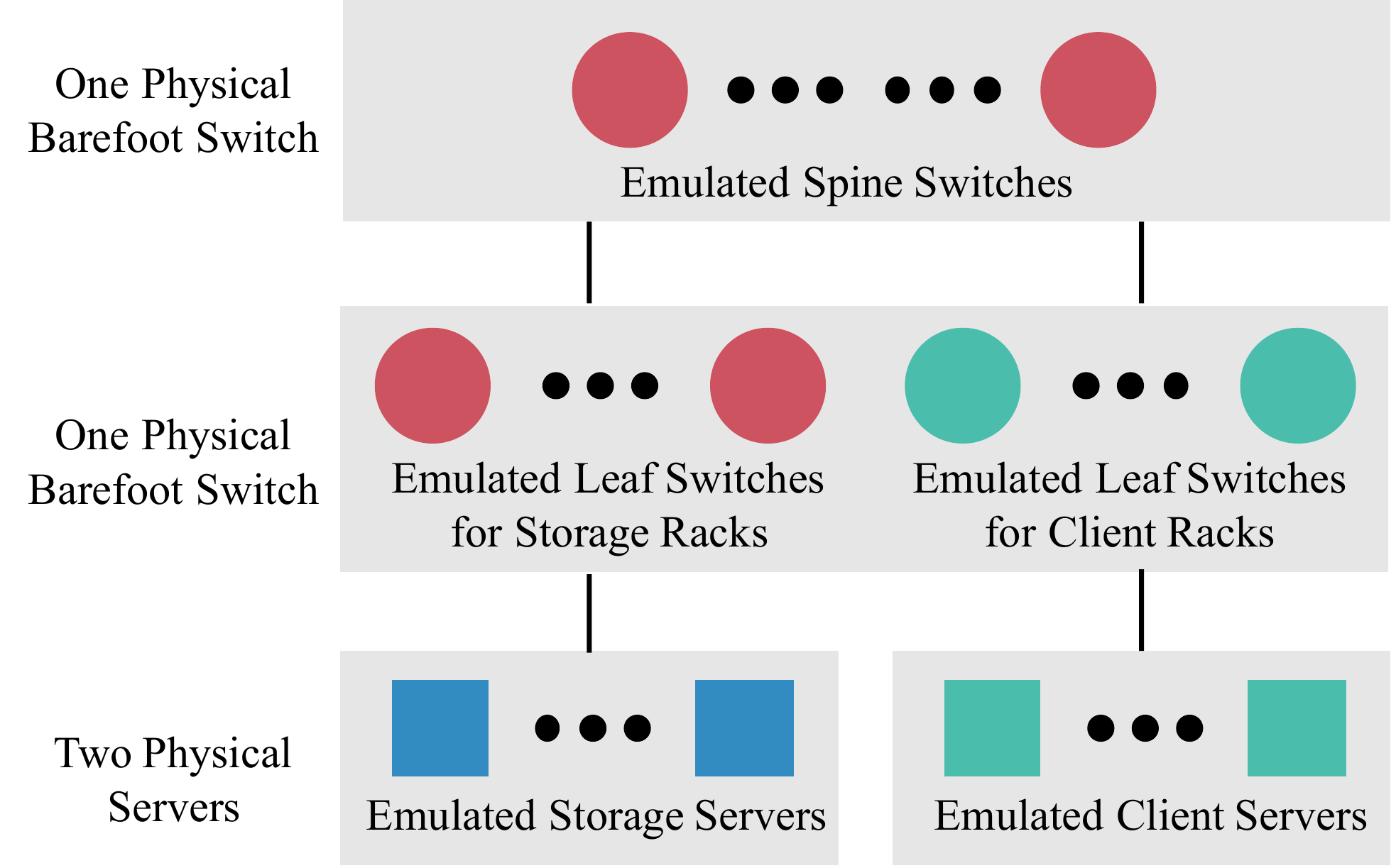}
\vspace{-0.25in}
\caption{Evaluation setup. The testbed emulates a datacenter with a two-layer leaf-spine
network by dividing switches and servers into multiple logical partitions.}
\vspace{-0.2in}
\label{fig:eval_testbed}
\end{figure}

\begin{figure*}[!t]
    \centering
    \subfigure[Throughput vs. skewness.]{
        \label{fig:eval_skew}
        \includegraphics[width=0.32\linewidth]{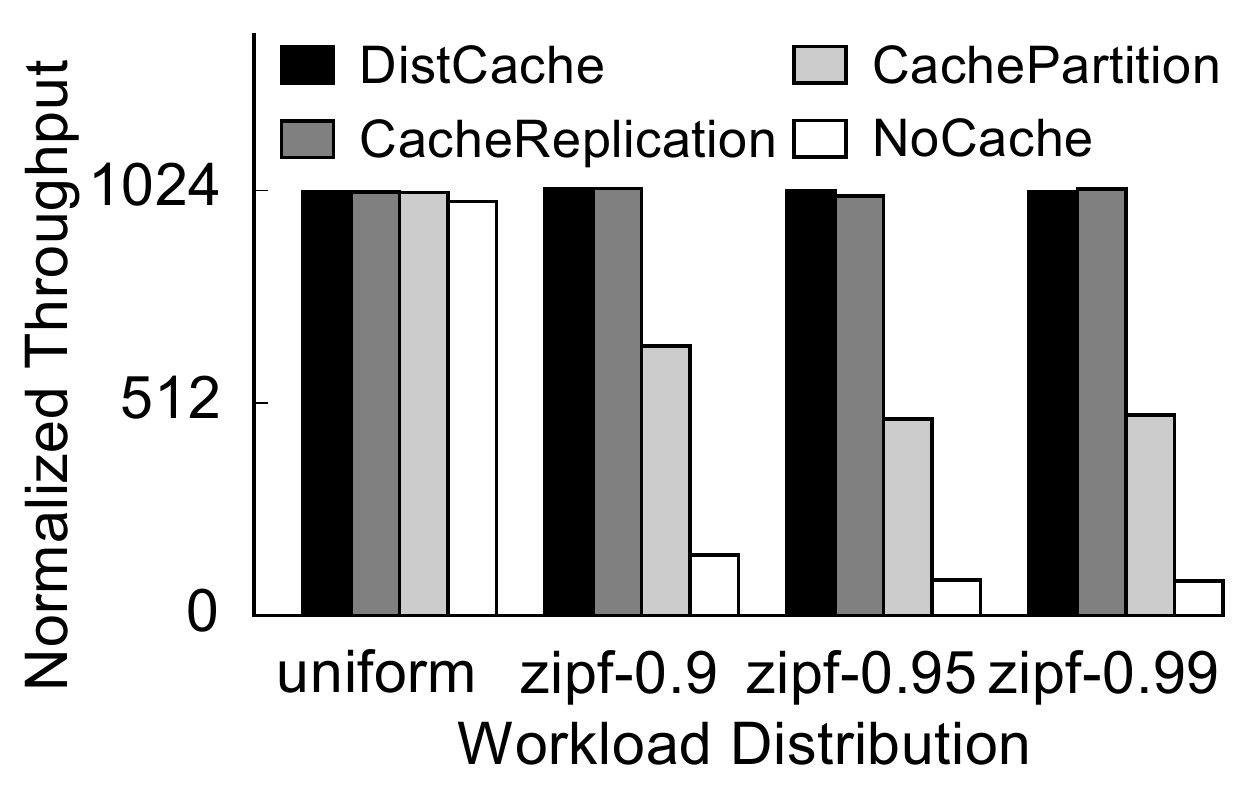}}
    \subfigure[Impact of cache size.]{
        \label{fig:eval_size}
        \includegraphics[width=0.32\linewidth]{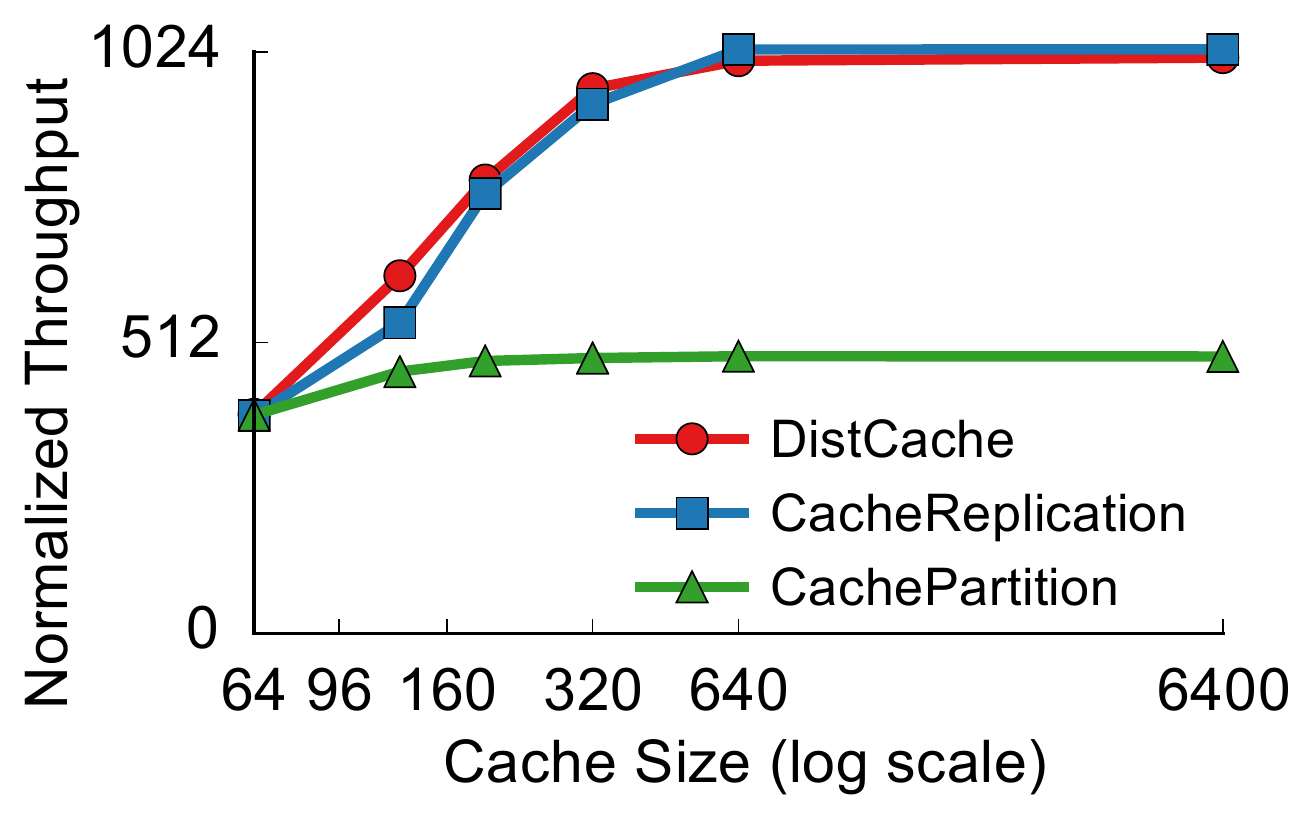}}
    \subfigure[Scalability.]{
        \label{fig:eval_scalability}
        \includegraphics[width=0.32\linewidth]{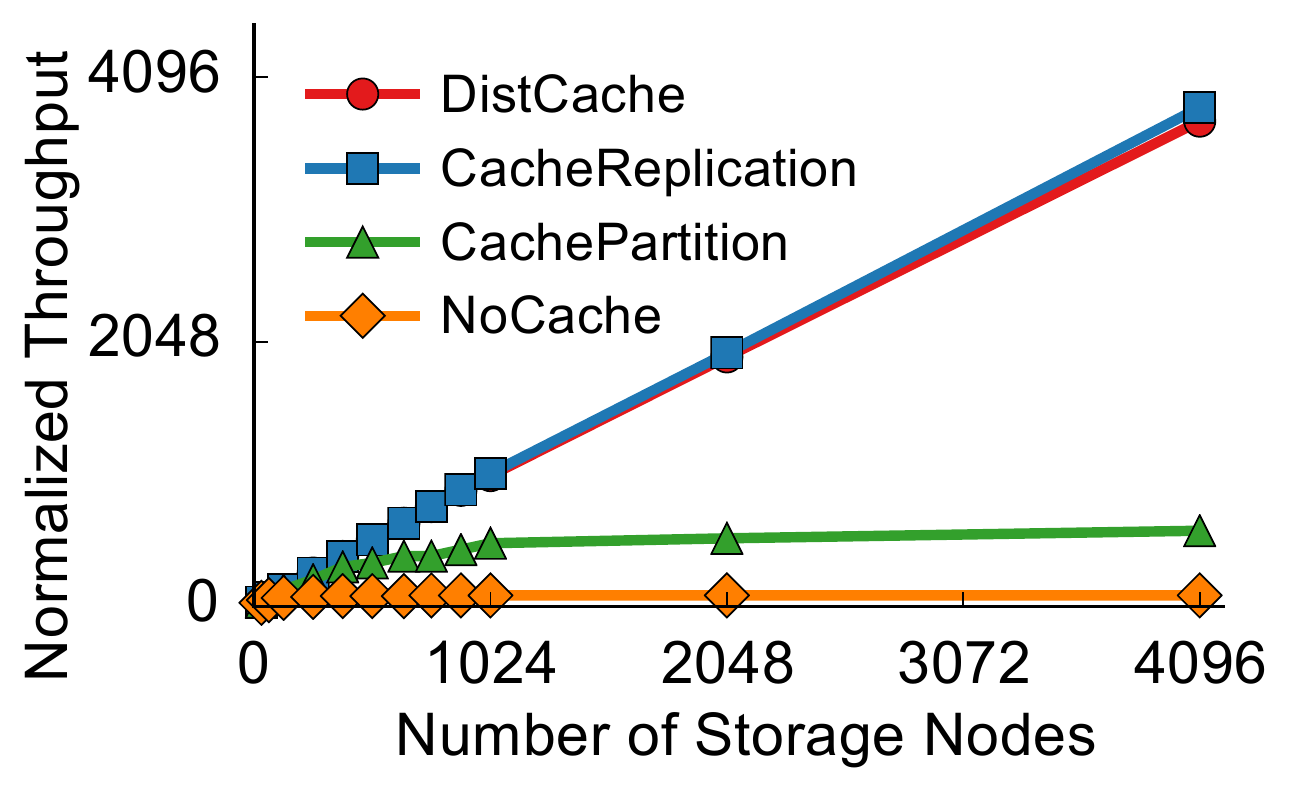}}
    \vspace{-0.15in}
    \caption{System performance for read-only workloads.}
    \vspace{-0.1in}
    \label{fig:eval_performance}
\end{figure*}

\subsection{Performance for Read-Only Workloads}
\label{sec:evaluation:performance}

We first evaluate the system performance of \sysname. By default, we use 32
spine switches and 32 storage racks. Each rack
contains 32 servers. We populate each cache switch with 100 hot objects, so that
64 cache switches provide a cache size of 6400 objects. We use read-only
workloads in this experiment, and show the impact of write queries in
\S\ref{sec:evaluation:coherence}. We vary workload skew, cache size and
system scale, and compare the throughputs of the four mechanisms under different
scenarios.

\para{Impact of workload skew.} Figure~\ref{fig:eval_skew} shows the throughput of the four
mechanisms under different workload skews. Under the uniform workload, the four
mechanisms have the same throughput, since the load between the servers is
balanced and all the servers achieve their maximum throughputs. However, when
the workload is skewed, the throughput of \nocache significantly decreases,
because of load imbalance. The more skewed the workload is, the lower throughput
\nocache achieves. \cpartition performs better than \nocache, by caching hot objects
in the switches. But its throughput is still limited because of load imbalance
between cache switches. \creplication provides the optimal throughput under
read-only workloads as it replicates hot objects in all spine switches. \sysname
provides comparable throughput to \creplication by using the distributed caching
mechanism. And we will show in \S\ref{sec:evaluation:coherence} that \sysname
performs better than \creplication under writes because of low overhead for
cache coherence.

\para{Impact of cache size.} Figure~\ref{fig:eval_size} shows the throughput of the
three
mechanisms under different cache sizes. \cpartition achieves higher throughput with
more objects in the cache. Because the skewed workload still causes load imbalance
between cache switches, the benefits of caching is limited for \cpartition. Some spine
switches quickly become overloaded after caching some objects. As such, the throughput
improvement is small for \cpartition. On the other hand, \creplication and \sysname
gain big improvements by caching more objects, as they do not have the
load imbalance problem between cache switches. The curves of \creplication and \sysname
become flat after they achieve the saturated throughput.

\para{Scalability.} Figure~\ref{fig:eval_scalability} shows how the four
mechanisms scale with the number of servers. \nocache does not scale because of
the load imbalance between servers. Its throughput stops to improve after a few
hundred servers, because the overloaded servers become the system bottleneck
under the skewed workload. \cpartition performs better than \nocache as it uses the
switches to absorb
queries to hot objects. However, since the load imbalance still exists between
the cache switches, the throughput of \cpartition stops to grow when there
are a significant number of racks. \creplication provides the optimal solution,
since replicating hot objects in all spine switches eliminates the load
imbalance problem. \sysname provides the same performance as
\creplication and scales out linearly.


\begin{figure}[!t]
    \centering
    \subfigure[Throughput vs. write ratio under Zipf-0.9 and cache size 640.]{
        \label{fig:eval_cache1}
        \includegraphics[width=0.9\linewidth]{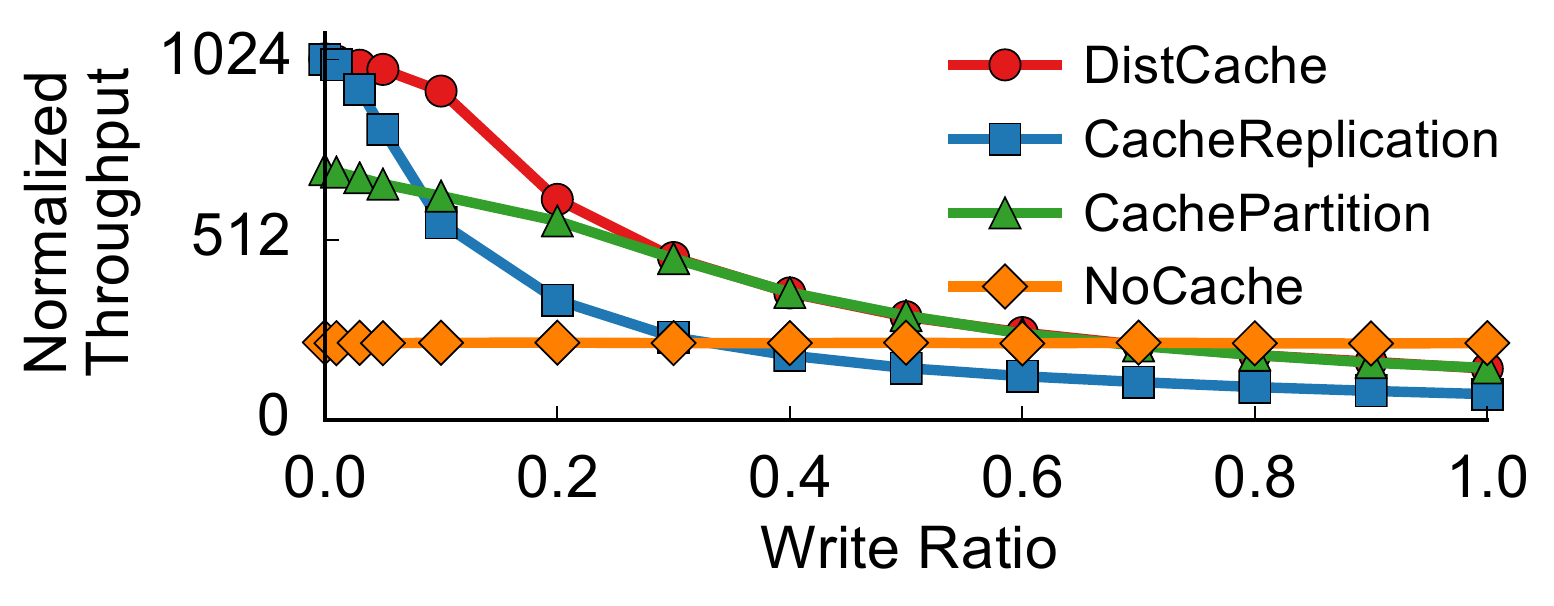}}
    \subfigure[Throughput vs. write ratio under Zipf-0.99 and cache size 6400.]{
        \label{fig:eval_cache2}
         \vspace{-0.1in}
        \includegraphics[width=0.9\linewidth]{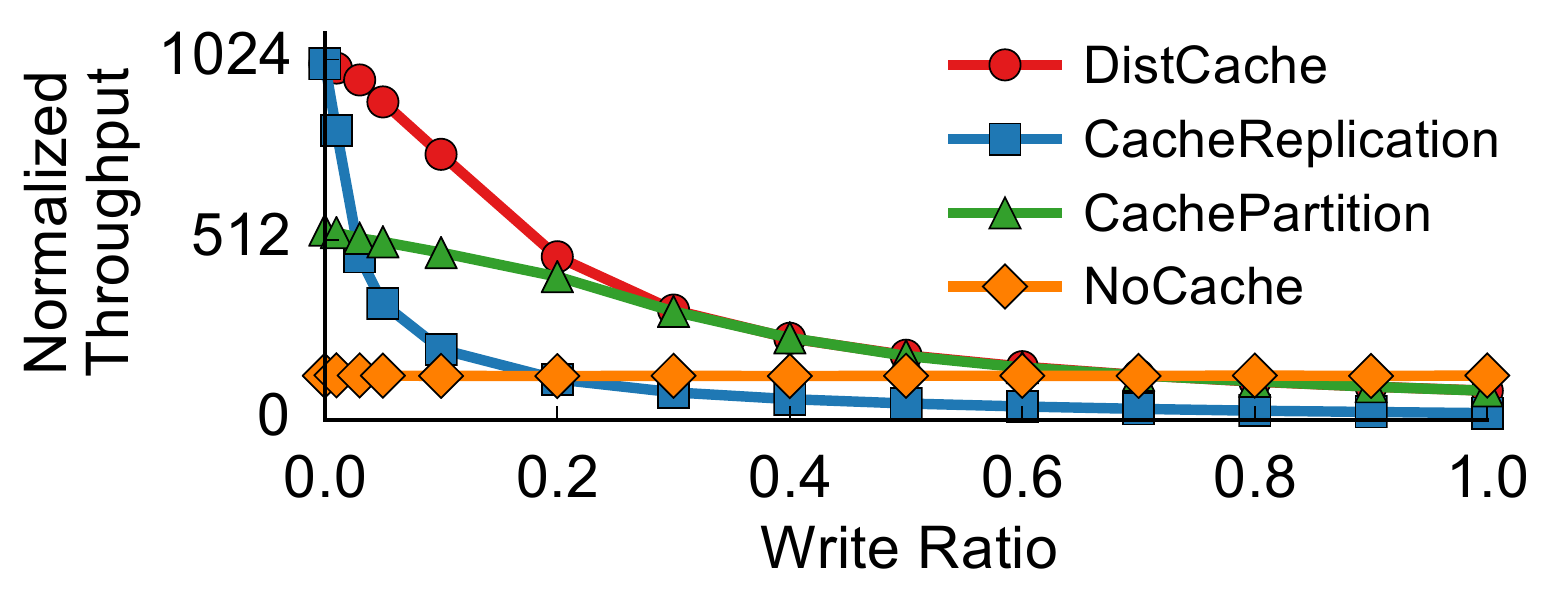}}
    \vspace{-0.15in}
    \caption{Cache coherence result.}
    \vspace{-0.30in}
    \label{fig:eval_cache}
\end{figure}

\subsection{Cache Coherence}
\label{sec:evaluation:coherence}

While read-only workloads provide a good benchmark to show the caching benefit,
real-world workloads are usually \emph{read-intensive}~\cite{workload-fb-sigmetrics12}. Write queries require the two-phase update
protocol to ensure cache coherence, which $(i)$ consumes
the processing power at storage servers, and $(ii)$ reduces the caching benefit
as the cache cannot serve queries to hot objects that are frequently being
updated. \creplication, while providing the optimal throughput under read-only
workloads, suffers from write queries, since a write query to a cached
object requires the system to update all spine switches. We use the basic setup
as the previous experiment, and vary the write ratio.

Since both the workload skew and the cache size would affect the result,
we show two representative scenarios. Figure~\ref{fig:eval_cache1} shows the scenario for Zipf-0.9
and cache size 640 (i.e., 10 objects in each cache switch). Figure ~\ref{fig:eval_cache2} shows
the scenario for Zipf-0.99 and cache size 6400 (i.e., 100 objects in each cache switch),
which is more skewed and caches more objects than the scenario in Figure~\ref{fig:eval_cache1}.
\nocache is not affected by the write ratio, as it does not cache anything (and our
rate limiter for the emulated storage servers assumes same overhead for read and write
queries, which is usually the case for small values in in-memory key-value
stores~\cite{mica}). The performance
of \creplication drops very quickly, and it
is highly affected by the workload skew and the cache size, as higher skewness and
bigger cache size mean more write queries would invoke the two-phase update protocol.
Since \sysname only caches an object once in each layer, it has minimal overhead for
cache coherence, and its throughput reduces slowly with the write ratio. The throughputs
of the three caching mechanisms eventually become smaller than that of \nocache,
since the servers spend extra resources on the cache coherence. Thus, in-network
caching should be disabled for write-intensive workloads, which is a general
guideline for many caching systems.

\begin{figure}[t]
\centering
\includegraphics[width=0.9\linewidth]{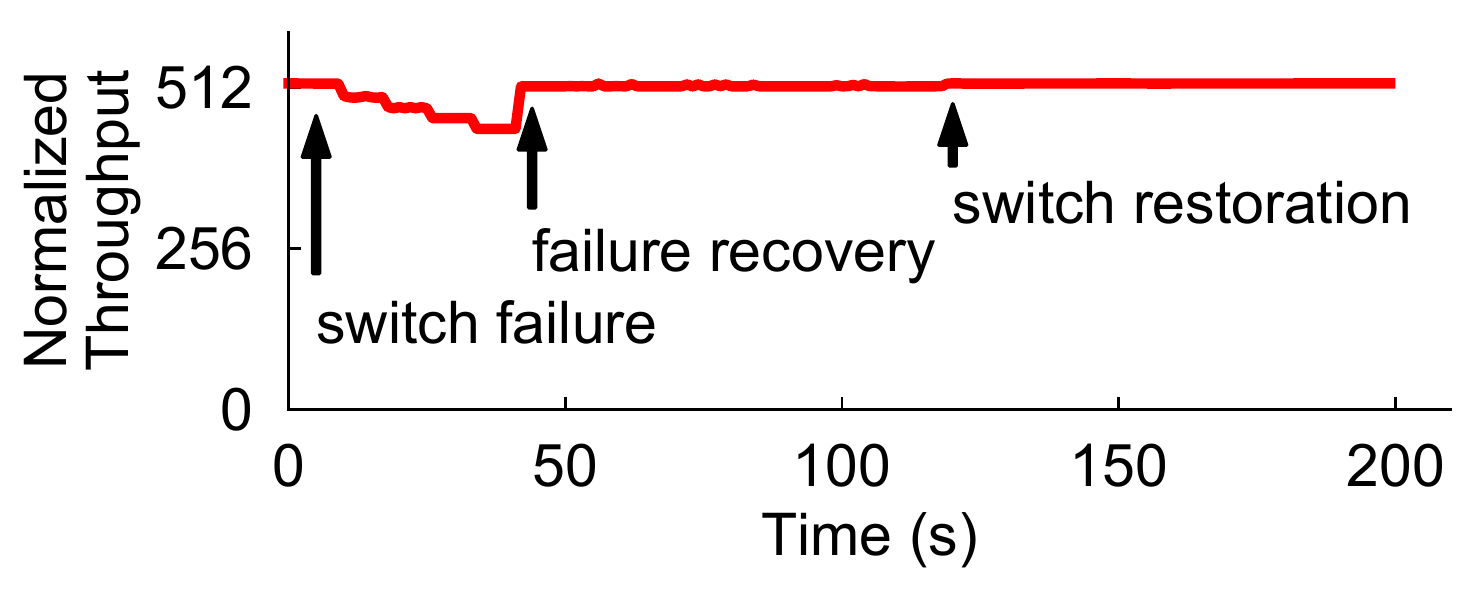}
\vspace{-0.2in}
\caption{Time series for failure handling.}

\label{fig:eval_failure}
\end{figure}


\subsection{Failure Handling}
\label{sec:evaluation:failure}

We now evaluate how \sysname handles failures. Figure~\ref{fig:eval_failure}
shows the time serious of this experiment, where x-axis denotes the time and
y-axis denotes the system throughput.
The system starts with 32 spine switches. We manually fail four
spine switches one by one. Since each spine switch provides 1/32 of the total
throughput, after we fail four spine switches, the system throughput drops to
about 87.5\% of its original throughput. Then the controller begins a failure
recovery process, by redistributing the partitions of the failed spine switches
to other alive spine switches. Since the maximum throughput the system can
provide drops to 87.5\% due to the four failed switches, the failure recovery
would have no impact if all alive spine switches were already saturated. To show
the benefit of the failure recovery, we limit the sending rate to half of the
maximum throughput. Therefore, after the failure recovery, the throughput can
increase to the original one. Finally, we bring the four failed switches back
online.

\subsection{Hardware Resources}
\label{sec:evaluation:resource}
Finally, we measure the resource usage of the switches. The programmable switches
we use allow developers to define their own packet formats and design the packet
actions by a series of match-action tables. These tables are mapped into
different stages in a sequential order, along with dedicated resources (e.g.,
match entries, hash bits, SRAMs, and action slots) for each stage. \sysname leverages
stateful memory to maintain the cached key-value items, and minimizes the
resource usage. Table~\ref{tab:resource} shows the resource usage
of the switches with the caching functionality. We show all the three roles,
including a spine switch, a leaf switch in a client rack, and a leaf switch in
a storage rack. Compared to the baseline Switch.p4, which is a fully functional
switch, adding caching only requires a small amount of resources, leaving plenty room for
other network functions.

\begin{table}[t]
\centering
\footnotesize
\begin{tabular}{ lccccc }
\toprule
 \textbf{Switches} & Match Entries & Hash Bits & SRAMs & Action Slots\\
\midrule
Switch.p4 &  804 &  1678 &  293 & 503 \\
Spine               & 149 &  751 & 250  & 98 \\
Leaf (Client)   &   76 &   209 & 91   &  32\\
Leaf (Server) &  120 &  721 & 252 & 108\\
\bottomrule
\end{tabular}
\vspace{-0.18in}
\caption{Hardware resource usage of \sysname.}
\vspace{-0.15in}
\label{tab:resource}
\end{table}

\section{Related Work}
\label{sec:related}

\paraf{Distributed storage.} Distributed storage systems are widely deployed to
power Internet services \cite{gfs, dynamo, haystack, memcache-nsdi13}. One
trend is to move storage from HDDs and SDDs to DRAMs for high
performance~\cite{memcached, redis, DAX, ramcloud}. Recent work has explored
both hardware solutions~\cite{blott2013achieving, chalamalasetti2013fpga,
lim2013thin, kvdirect, kalia2014using, rdma-atc16, farm, li2011system,
lotfi2012scale, gutierrez2014integrated, novakovic2014scale, li2015architecting}
and software optimizations~\cite{fawn, li2014algorithmic, mica, memc3,
Vasudevan12, silt}. Most of these techniques focus on the
single-node performance and are orthogonal to \sysname, as \sysname focuses on
the entire system spanning many clusters.

\para{Load balancing.} Achieving load balancing is critical to scale out
distributed storage. Basic data replication techniques~\cite{consistent-hashing,
virtual-nodes} unnecessarily waste storage capacity under skewed workloads.
Selective replication and data migration
techniques~\cite{adaptive-caching-eurosys15, yahoo-load-balancer, estore}, while
reducing storage overhead, increase system complexity and performance overhead
for query routing and data consistency. EC-Cache~\cite{eccache} leverages
erasure coding, but since it requires to split an object into multiple chunks,
it is more suitable for large objects in data-intensive applications. Caching is
an effective alternative for load balancing~\cite{Fan:smallcache, switchkv,
netcache}. \sysname pushes the caching idea further by introducing a distributed
caching mechanism to provide load balance for large-scale storage systems.

\para{In-network computing.} Emerging programmable network
devices enable many new in-network applications.
IncBricks~\cite{incbricks-asplos17} uses NPUs as a key-value cache.
It does not
focus on load balancing. NetPaxos~\cite{netpaxos, netpaxos-ccr} presents a
solution to implement Paxos on switches. SpecPaxos~\cite{specpaxos}
and NOPaxos~\cite{nopaxos} use switches to order messages to
improve consensus protocols. Eris~\cite{eris} moves
concurrency control to switches to improve distributed
transactions.

\section{Conclusion}
\label{sec:conclusion}

We present \sysname, a new distributed caching mechanism for large-scale storage
systems. \sysname leverages independent hash functions for cache allocation and
the power-of-two-choices for query routing, to enable a ``one big cache''
abstraction. We show that combining these two techniques provides provable load
balancing that can be applied to various scenarios. We demonstrate
the benefits of \sysname by the design, implementation and evaluation of the
use case for emerging switch-based caching.

\para{Acknowledgments}\hspace{0.1in}
We thank our shepherd Ken Salem and the reviewers for their valuable feedback.
Liu, Braverman, and Jin are supported in part by NSF grants CNS-1813487, CRII-1755646 and CAREER 1652257, 
Facebook Communications \& Networking Research Award, Cisco Faculty Award, ONR Award N00014-18-1-2364, DARPA/ARO Award W911NF1820267, and Amazon
AWS Cloud Credits for Research Program. Ion Stoica is supported in part by NSF
CISE Expeditions Award CCF-1730628, and gifts from Alibaba, Amazon Web Services,
Ant Financial, Arm, CapitalOne, Ericsson, Facebook, Google, Huawei, Intel,
Microsoft, Scotiabank, Splunk and VMware.

{
\bibliographystyle{ieeetr}
\bibliography{ref}}
\newpage
\appendix
\section{Analysis of our algorithm}\label{sec:proof}
This section provides formal proofs for the lemmas and theorem
presented in Section~3 of the paper.

\subsection{Recap of our models and comparison}
We start with reviewing notations and models setup in this paper. Recall that we have
a total number of $k$ distinct objects. The arrival rate for the $i$-th object is
$p_i$.  We have a total number of $2m$ cache nodes and the processing rate for each cache node is $\widetilde T$. Our goal is to show that our PoT algorithm is stationary over the long run.

\parait{The bipartite graph $G=(U,V,E)$.} Recall that $U = \{o_1, o_2, \dots, o_
{k - 1}\}$ is the set of objects. $A = \{a_0, \dots, a_{m - 1}\}$ and $B = \{b_0, \dots b_{m - 1}\}$ are two collections of cache nodes. Let $V = A \cup B$. We build a bipartite graph $G$ in which the left-hand side is $U$ and the right-hand side is $V$. The edge set is
\begin{equation}
E = \{\{o_i, a_{j_0} \}\mid h_0(i) = j_0\} \cup \{\{o_i, b_{j_1}\} \mid h_1(i) = j_1\}
\end{equation}

Also, let $\Gamma(v)$ be the set of $v$'s neighbors in $G$ and $\Gamma(S) = \cup_{v \in S} \Gamma(v)$ for any subset $S$ of nodes.

\para{Comparison to the balls-and-bins model.} Despite the apparent similarity between
the balls-and-bins model, our process is substantially different. The standard techniques used to analyze PoT algorithms are  not directly applicable in our model.

\parait{Recap of the PoT algorithm in the balls-and-bins model.} In the basic setting,
there
are a total number of $m$ bins. At each step, a new ball arrives. We use the PoT rule to determine a bin to store the new ball: we uniformly choose two random bins and place the ball to the bin with lighter load. After that, a deletion event happens: we uniformly choose a bin and remove a ball from the bin (but we do nothing if the bin is already empty). We are primarily interested in the maximum load of the bins in the stationary state.

Extensive studies of the balls-and-bins model and its ``sensible variations'' such
as the continue time model~\cite{luczak2005power,luczak2006maximum,bramson2010randomized,bramson2012asymptotic},
the adversarial model~\cite{cole1998balls,azar1999balanced}, etc. have determined
that the max load of the bin is $O(\log \log n)$, in contrast to the max load of $\Theta(\log n)$ for uniform balls-in-bins algorithm (i.e., place a new ball in a randomly chosen bin).

\parait{Main difference.} The main difference between our process and the PoT balls-and-bins
process is that ours only has a total number of $k$ types of objects. Objects of the
same type will need to ``reuse'' the hash functions, whereas in the balls-and-bins
process, new random sources are used to choose the bins, regardless of the history (i.e., each new ball will have fresh randomness). With the new randomness, it is easier to argue that the random choices will not be bad for an extensive period of time. On the other hand, in our process we need to argue that by using only two hash functions, the system will be stable even when the number of queries is $\gg k$. In fact, the performance gap between the PoT algorithm and the uniform algorithm highlights the distinction between our process and the balls-and-bins process (proved in Section~\ref{asec:lower}).

\begin{lemma}\label{lem:lower} Using the notations above, with constant probability, our system is non-stationary when we use the uniform algorithm; with $o(1)$ probability, the system is stationary when we use the PoT algorithm.
\end{lemma}

In other words, the PoT algorithm makes a ``life-or-death'' improvement instead of ``shaving off a $\log n$'' improvement: our system is provably unreliable without the PoT algorithm.

\para{Outline of the proof.} We may think of our process as a ``flow problem'', in
which we build a bipartite graph, where the left-hand side is the objects and the
right-hand side is the cache nodes. An object connects to a cache node if and only
if one of $h_0$ and $h_1$ maps the object to the cache node.

When a request for an object is handled by a cache node, there is a unit flow moving from the object to the cache node. Therefore, the objects correspond to the source/supply nodes and the cache nodes correspond to the sink nodes.

At a high level, our analysis aims to show that using the PoT algorithm in our process routes the requests according to a feasible flow. Our analysis consists of two steps.

\parait{Step 1. Show that a feasible flow exists (Section~\ref{asec:feasible}).}
A necessary condition for the PoT algorithm to work is that a feasible solution exists for the bipartite graph flow problem. When a feasible solution does not exist, no algorithm is able to produce a stabilized system. One can also see that a feasible flow also has a natural matching interpretation (i.e., how requests are matched to the cache nodes). Therefore, we use feasible flow and perfect matching interchangeably in the rest of the analysis.

\parait{Step 2. Show that the PoT algorithm ``implements'' a feasible flow (Section~
\ref{asec:flowimpliespot}).} Building a feasible flow requires that the nodes (corresponding to objects) on the left-hand side of the bipartite graph can intelligently split the flow. Instead of performing a global computation to find a feasible flow, Step 2 demonstrates how the PoT policy, which is essentially a local algorithm, automatically finds a feasible solution.

The sections below explain each of the steps in detail.

\subsection{Feasible flows/matching exists}\label{asec:feasible}
This section explains how a perfect matching exists in $G$. Recall the definition of perfect matching:

\begin{definition}[Repeat of Definition~1 in the paper.] Let $\Gamma(v)$ be the set of neighbors of vertex $v$ in $G$.
A weight assignment $W = \{w_{i,j} \in [0, \widetilde{T}] |
e_{i,j} \in E\}$
is a perfect
matching of $G$ if
\begin{enumerate}
    \item $\forall o_i \in U: \quad \sum_{v \in \Gamma(o_i)} w_{o_i,v} = p_i \cdot R$, and
    \item $\forall v \in V: \quad \sum_{u \in \Gamma(v)} w_{u,v} \leq \widetilde{T}$.
\end{enumerate}
\label{def:matching}
\end{definition}

We aim to prove Lemma~\ref{lem:matching}, i.e.,

\begin{lemma}[Repeat of Lemma~1 in the paper]
Let $\alpha$ be a suitably small constant. If $k \leq m^\beta$ for some constant
$\beta$ (i.e., $k$ and $m$ are polynomial-related)
and $\max_i(p_i) \cdot R \leq \widetilde{T}/2$,
then for any $\epsilon > 0$, there exists a perfect matching
for $R = (1-\epsilon)\alpha \cdot m\widetilde{T}$ and any $P$, with high probability
for sufficiently
large $m$.
\label{lem:matching}
\end{lemma}

\para{Techniques and roadmap.} We develop new techniques to marry random/expansion
graph theory with primal-dual properties for flows.  First, we show that random bipartite graph has the so-called ``expansion properties''.  Then, we show that a lower bound exists on the optimal flow size by using the expansion property in the dual of the flow problem.

We decompose our analysis into two steps: \emph{Step 1a:} we analyze a special case when $k = \alpha m$ and the request rates $p_i$ are uniform (Section~\ref{asec:uniform}), and \emph{Step 1b:} we will reduce the remaining cases to the special case (Section~\ref{asec:generalized}).

\subsubsection{Step 1a: when $k = \alpha m$ and $p_i$'s are uniform}\label{asec:uniform}

This section proves the following lemma.
\begin{lemma}Let $\alpha$ be a suitably small constant. Let $k = \alpha m$,  $R = (1-\epsilon) \alpha m \widetilde T$, and $p_i = \frac{(1-\epsilon)\widetilde T}{R}$ for all $i$. There exists a perfect matching for $G$ with high probability.
\end{lemma}\label{lem:specialmatching}

We note that \emph{(i)}  Lemma~\ref{lem:matching} requires $\max_i(p_i) \cdot R \leq \widetilde{T}/2$ but in Lemma~\ref{lem:specialmatching} we have $\max_i(p_i) \cdot R = (1-\epsilon) \widetilde{T}$. The constant term in Lemma~\ref{lem:matching} is worse because the reduction in
Section~\ref{asec:generalized} will introduce a constant factor loss, and \emph{(ii)} because $p_i$'s all have the same value, we can  view $G$ as an unweighted graph (however, fractional solutions are still acceptable).

We shall show that $G$ possesses the so-called ``expansion property'' (a notation borrowed from the spectral graph theory), and the expansion property implies the
existence of a perfect matching.

\begin{definition} A bipartite graph $G = (U, V, E)$ has the expansion property if for any $S \subseteq U$, we have
\begin{equation}
|\Gamma(S)| \geq |S|.
\end{equation}
\end{definition}

We recall the Hall's theorem:

\begin{theorem}\label{lem:expansionimpliesmatching}
\cite{bollobas2013modern} Let $G = (U, V, E)$ be an arbitrary unweighted bipartite graph. If $G$ has the expansion property, there exists a perfect matching\footnote{Here, we use the standard definition of perfect matching for unweighted bipartite graph, instead of Definition~\ref{def:matching}.} in $G$.
\end{theorem}

Theorem~\ref{lem:expansionimpliesmatching} implies that to prove Lemma~\ref{lem:specialmatching}, we only need to show that $G$ has the expansion property.

\begin{lemma}\label{lem:expansion} Using the notations above, with diminishing probability (i.e., $o(1)$), $G$ does not have the expansion property.
\end{lemma}

\begin{proof} Our goal is to give an upper bound for
\begin{equation}
\Pr[\exists S \subseteq U: |\Gamma(S)| < |S|] \leq \sum_{S\subseteq U} \Pr[|\Gamma(S)| < |S|].
\end{equation}
We divide the terms in the right-hand-side of the above inequality into two groups, each of which exhibits different combinatorics properties. Therefore, we need to use different techniques to bound the sums of the terms in these two groups.

\para{Group 1:} $|S| \geq m^{0.1}$. We aim to bound:

\begin{equation}
\sum_{S \subseteq V, |S| \geq m^{0.1}}\Pr[|\Gamma(S)| < |S|] = \sum_{L \geq m^{0.1}}\sum_{|S| = L} \Pr[|\Gamma(S)| < L]
 \end{equation}

 Fix $L$ and let $S \subseteq V$ such that $|S| = L$. Let the out-going edges of $S$ be $e_1, e_2, ..., e_{2L}$. Define an indicator random variable $\cale_i$ that sets to 1 if and only if the right-end of $e_i$ coincides with the right-end of an $e_j$ for some $j < i$. We refer to $\cale_i$ as a ``repeat'' event. Note that $|\Gamma(S)| < L$ if and only if there exists at least $L + 1$ repeats. Therefore,
\begin{eqnarray*}
\Pr[|\Gamma(S)| < L] & \leq & \Pr[\sum_{i \leq 2L}\cale_i \geq L + 1] \\
& \leq & \Pr[\sum_{i \leq 2L}\cale_i \geq L] \\
& \leq & \binom{2L}{L}\left(\frac{L}{m}\right)^L \\
& \leq & \left(\frac{2eL}{m}\right)^{L}
\end{eqnarray*}
The last inequality  uses $\binom{a}{b} \leq \left(\frac{ae}{b}\right)^b$ for any integers $a$ and $b$.
Now, we can bound the sum of all of the terms in Group 1.
\begin{eqnarray*}
& & \sum_{L \geq  m^{0.1}}\sum_{|S| = L}\Pr[|\Gamma(S)| < L] \\
& \leq & \sum_{L \geq m^{0.1}}\binom{k}{L}\left(\frac{2eL}{M}\right)^L \\
& \leq & m\left(\frac{\alpha m e}{L}\right)^L\left(\frac{2eL}{m}\right)^{L} \\
& \leq & m (2e^2\alpha)^{m^{0.1}} \leq 1/m^{10}
\end{eqnarray*}
for sufficiently large $m$ so long as $\alpha > 1/(2e^2)$.

\para{Group 2.} $|S| < m^{0.1}$. We aim to bound:
\begin{equation}
\sum_{L < m^{0.1}} \sum_{|S| = L}\Pr[|\Gamma(S)| < k].
\end{equation}

Define $\ell_1 \triangleq |\Gamma(S) \cap A|$ and $\ell_2  \triangleq |\Gamma(S) \cap B|$ (recall that $A$ and $B$ are two groups of cache nodes). One can see that $|\Gamma(S)| = \ell_1 + \ell_2$. Therefore, a necessary condition for $|\Gamma(S)| < |S|$ is $\ell_1 + \ell_2 \leq |S| - 1$. For any $S \subseteq U$, we have

\begin{eqnarray*}
\Pr[|\Gamma(S)| < |S|]  & \leq &  \sum_{\ell_1 + \ell_2 \leq L - 1} \binom{m}{\ell_1}\left(\frac{\ell_1}{m}\right)^{L} \binom{m}{\ell_2}\left(\frac{\ell_2}{m}\right)^{L} \\
&  \leq & \sum_{\ell_1 + \ell_2 \leq L - 1}\frac{m^{\ell_1 + \ell_2}e^{\ell_1 + \ell_2}(\ell_1\ell_2)^{L}}{m^{2L}\ell_1^{\ell_1}\ell_2^{\ell_2}}.
\end{eqnarray*}

Next, we have
\begin{eqnarray*}
& & \sum_{L \leq m^{0.1}}\sum_{|S| = L}\Pr[|\Gamma(S)| < |S|] \\
& \leq & \sum_{L \leq m^{0.1}}\binom{\alpha m}{L}\sum_{\ell_1 + \ell_2 \leq L - 1}\frac{m^{\ell_1+\ell_2}e^{\ell_1 + \ell_2}(\ell_1\ell_2)^L}{m^{2L}\ell_1^{\ell_1}\ell_2^{\ell_2}} \\
& \leq & \sum_{L \leq m^{0.1}} \sum_{\ell_1 + \ell_2 \leq L - 1}\frac{(\alpha e^2)^Lm^{\ell_1+\ell_2}}{(Lm)^L}\cdot \frac{(\ell_1\ell_2)^L}{\ell_1^{\ell_1}\ell_2^{\ell2}} \\
& \leq & \sum_{L < m^{0.1}}L \max_{\ell_1 + \ell_2 \leq L}\left\{\frac{(\alpha e^2)^2 m^{\ell_1+\ell_2}(\ell_1\ell_2)^L}{(Lm)^L\ell_1^{\ell_1}\ell_2^{\ell_2}}\right\}
\end{eqnarray*}
Next, we find an upper bound on
\begin{equation}
G(\ell_1, \ell_2) \triangleq  \left\{\frac{(\alpha e^2)^2 m^{\ell_1+\ell_2}(\ell_1\ell_2)^L}{(Lm)^L\ell_1^{\ell_1}\ell_2^{\ell_2}}\right\}.
\end{equation}
First, note that when $\ell_1 + \ell_2 < L - 1$, we can set $\ell_1 \leftarrow \ell_1 + 1$ to increase $G(\ell_1, \ell_2)$.

Second, we may compute $\frac{d}{d \ell_1}\log(G(\ell_1, L - 1 - \ell_1))$ to find the optimal $\ell_1$. We have
\begin{eqnarray*}
 & & \frac{d}{d \ell_1}\log G(\ell_1, L - 1 - \ell_1) \\
 & = &\frac{L}{(L - 1 - \ell_1)\ell_1} -2 + \log(L - 1 - \ell_1) - \log(\ell_1).
\end{eqnarray*}
We can see that the optimal value is achieved when $\ell_1 = c_1 L$ and $\ell_2 = c_2 L$ for some $c_1, c_2 = \Theta(1)$ and $c_1 + c_2 = \frac{L - 1}L$. Thus,
\begin{eqnarray*}
& &  \sum_{L < m^{0.1}}L \max_{\ell_1 + \ell_2 \leq L}\left\{\frac{(\alpha e^2)^2 m^{\ell_1+\ell_2}(\ell_1\ell_2)^L}{(Lm)^L\ell_1^{\ell_1}\ell_2^{\ell_2}}\right\} \\
& \leq & \sum_{L \leq m^{0.1}}L\frac{(\alpha e^2)^L}{m}\frac{(c_1c_2)^LL^{L}}{(c_1^{c_1}c_2^{c_2}L)^{L-1}} \\
&  & \quad \mbox{(uses $c_1 + c_2 = \frac{L - 1}{L}$)} \\
&\leq  & \sum_{L \leq m^{0.1}}L\left(\frac{\alpha e^2c_1c_2}{c_1^{c_1}c_2^{c_2}}\right)^{L}\frac L m \leq \frac{1}{m^{0.7}}.
\end{eqnarray*}

The last inequality holds when $\alpha \leq \frac{c_1^{c_1}c_2^{c_2}}{e^2c_1c_2}$.
\end{proof}

Finally, Lemma~\ref{lem:expansion} and Lemma~\ref{lem:expansionimpliesmatching} imply Lemma~\ref{lem:specialmatching}.

\subsubsection{Step 1b. Generalization}\label{asec:generalized}
This section tackles the more general case, in which no constraints are imposed on $k$ and $p_i$. We first explain the intuition why the general cases can be reduced to the uniform case analyzed in Section~\ref{asec:uniform}:

\parait{Intuition part 1 (IP1): when $k < \alpha m$.} In this case, $p_i$ is still
constrained by $\max_i p_i \cdot R \leq \widetilde{T} /2$ and therefore $p_i = O(\widetilde T / R)$. Thus, the new problem is equivalent to deleting one or more objects in the special case, and is a ``strictly easier problem'' (i.e., the requests are strictly smaller).

\parait{Intuition part 2 (IP2): when $k > \alpha m$.}
Note that when $k$ is larger than $\alpha m$, the total rate $R$ remains unchanged. This intuitively corresponds to splitting some objects into smaller ones. Consider, for example, $o_{\alpha m}$ is the splitted into $o_{\alpha m}$ and $o_{\alpha m + 1}$ so that the new request rates are halved $p_{\alpha m} = p_{\alpha m + 1} = \frac{p_1} 2$. In this new problem $k = \alpha  m + 1$. Originally, there were only two hash functions handling $o_{\alpha m}$. After the splitting, four hash functions handle the same amount of requests.
Note that load-balance improves when there are more hash functions.


The above intuitions assume the requests are mostly homogeneous. The most challenging case is when $k$ is large and $p_i$'s have a long tail. This can be viewed as a ``mixture'' of the above two cases: approximately $\alpha m$ objects have large $p_i$ (resembling the scenario addressed in IP1), and the rest of the objects have small $p_i$ (resembling the case in IP2). Our main technique here is to decompose  a  problem with heterogeneous  $p_i$'s into subproblems so that  each of them has homogeneous $p_i$'s. Then we argue with high probability perfect matchings exist for all the sub-problems.

Below we first describe/define basic building blocks for constructing large systems (i.e., sub-problems we are able to solve). Then we explain our decomposition analysis.

Central to our analysis is our introduction of $(\alpha, \gamma, t)$-graph family. This definition enables us to scale up (increase $k$ and/or $m$) or scale out (increase throughputs) of the basic building blocks. We say

\begin{definition}
 $\{G_i\}_i$ is a $(\alpha, \gamma, t)$-graph family if the graph $G_i = (U_i, V_i,
 E_i)$ is a bipartite graph that represents the relationship between objects and cache
 nodes such that (i) there are $i$ cache nodes in each of $A$ and $B$, (ii) $|U_i|
 = \alpha i$, (iii) all the $p_j$'s are the same, (iv) $R \leq \gamma(1-\epsilon)m t$ for some small constant $\epsilon$; and (v) the cache node's capacity is $t$.

A $(\alpha, \gamma, t)$-graph family is \emph{$(1-\delta)$-feasible} if for sufficiently large $i$, with probability $\delta$ a perfect matching exists.
\end{definition}

\para{Example.} Lemma~\ref{lem:specialmatching} shows that the $(\alpha, \alpha, \widetilde T)$-graph family is $(1-\frac{1}{m^{0.7}})$-feasible.

Different graph families may be considered as ``basic building blocks''', and we may put multiple graph families together to construct new systems that have perfect matching.

For example, suppose $(\alpha, \alpha, 1)$ and $(2\alpha, \alpha, 1)$ graph families ($\{G^{(1)}\}_i$ and $\{G^{(2)}\}_i$, respectively) are both $(1-\delta)$-feasible. In this case, the request rates for objects in $\{G^{(1)}\}_i$ are $1-\epsilon$ whereas the request rates in  $\{G^{(2)}\}_i$ are $0.5(1-\epsilon)$.
We can define an addition operation between these two families through coupling.

\parait{Coupling.} There is a natural way to couple the two graphs $G^{(1)}_m$ and
$G^{(2)}_m$. We may imagine that $G^{(1)}_m$ has $2\alpha m$ objects but the last $\alpha m$ objects have zero request rates. This allows us to couple the hash functions used in two graphs in the natural way.

Under this coupling, we can define the addition operation on $G^{(1)}_m$ and $G^{
(i)}_m$. In the new system, the total number of cache nodes remains unchanged ($2m$ cache nodes). Each cache node is a ``merge'' of two cache nodes (one from $G^{(1)}_m$ and the other from $G^{(2)}_m$). So the new capacity is 2. There are $2\alpha m$ jobs. The request rates of half of them are $1.5(1-\epsilon)$. The other half are $0.5(1-\epsilon)$. Because $G^{(1)}_m$ and $G^{(2)}_m$ are all $1-\delta$ feasible, it becomes straightforward to compute the probability a perfect matching exists in the new system:
using a union bound, this probability is $1-2\delta$.

We next walk through fundamental properties of $(\alpha, \gamma, t)$-graph families. Specifically, we show that the $(\alpha, \gamma, t)$-graph family is $(1-\delta)$-feasible as $\alpha$, $\gamma$, and $t$ scale in a suitable manner.

First, we observe that $t$ is a scale-free parameter.

\begin{lemma} If  $(\alpha, \gamma, 1)$-graph is $(1-\delta)$-feasible, then $(\alpha, \gamma, t)$-graph is $(1-\delta)$-feasible for any $t$.
\end{lemma}

By using a similar technique for proving Lemma~\ref{lem:specialmatching}, we also have

\begin{lemma}\label{lemma:smallblocks} There exists a suitably small constant $\alpha$ such that $(\rho \alpha, \alpha, t)$-graph families are all $(1-\frac{2}{m^{0.7}})$-feasible for any $\rho \in [1,2]$.
\end{lemma}

Then we also show that when we have more objects, each of which has lower request rates, we will have a greater chance to see a perfect matching.

\begin{lemma}\label{lemma:largeblocks} If $(\alpha, \gamma, t)$-graph family is $(1-\delta)$-feasible, then $(k\alpha, \gamma, t)$-graph family is also $(1-\delta)$-feasible for any positive integer $k$.\footnote{
Here, $k$ does not refer to the number of objects.}
\end{lemma}
\begin{proof}
The intuition of the proof is that we may view a cache node whose processing rate is $t$ as being no worse than a union of $k$ cache nodes, each of which has processing rate $t/k$. Specifically, we  couple the following processes:

\parait{Process $G^{(1)}$:} There are a total number of $\alpha m$ objects.  Each
of $A$ and $B$ consists of $m$ cache nodes with processing rate $t$. The request rate for each object is $(1-\epsilon)\gamma t/\alpha$.

\parait{Process $G^{(2)}$:} There are a total number of $\alpha k m$ objects. Each
of $A$ and $B$ consists of $km$ cache nodes with processing rate $t/k$. The request rate for each object is $(1-\epsilon)\gamma t/(\alpha k)$.

The graph $G^{(2)}$ is in the $(\alpha, \gamma, t)$-family, so based on the assumption made in the lemma, with probability $(1-\delta)$ a perfect matching exists.

We can also couple $G^{(1)}$ and $G^{(2)}$ in a way that if there exists a perfect
matching in $G^{(2)}$, then there exists a perfect matching in $G^{(1)}$. Specifically, we may imagine that each cache node in $G^{(1)}$ consists of $k$ smaller cache nodes in $G^{(2)}$. Let $h^{(1)}_i$ ($h^{(2)}_i$) be the hash functions used in $G^{(1)}$ ($G^{(2)}$). We shall couple $h^{(1)}_i$ and $h^{(2)}_i$ in the following way:
\begin{equation}
h^{(1)}_i(j) = \lfloor h^{(2)}_i(j)  / k \rfloor.
\end{equation}
By tying the hash functions in this way, a perfect matching in $G^{(2)}$ can be reduced to a perfect matching in $G^{(1)}$.

\end{proof}

We use Lemma~\ref{lemma:smallblocks} and Lemma~\ref{lemma:largeblocks} to prove a more general result.

\begin{lemma}\label{lem:fullblock} There exists an $\alpha$ such that the $(\rho \alpha, \alpha, t)$-graph families are $(1-2/m^{0.7})$-feasible for all $\rho \geq 1$.
\end{lemma}

\begin{proof} Let $\rho' = \rho / \lfloor \rho \rfloor \in [1,2]$. By Lemma~\ref{lemma:smallblocks}, the $(\rho'\alpha, \alpha, t)$-graph family is $(1-2/m^{0.7})$-feasible. Then By Lemma~\ref{lemma:largeblocks}, the
$(\rho' \lfloor \rho \rfloor \alpha, \alpha, t)$-graph family is also $(1-2/m^{0.7})$-feasible for any $\rho$, which proves the lemma.
\end{proof}

We are now ready to prove Lemma~\ref{lem:matching}.
\begin{proof}[Proof of Lemma~\ref{lem:matching}]

At a high level, our goal is to decompose the matching problem into multiple smaller matching problems, each of which can be shown to have a perfect matching by using Lemmas~\ref{lem:specialmatching} or~\ref{lem:fullblock}.

Our analysis consists of two steps:
\begin{itemize}
\item \emph{Step 1. Rounding up.}  Our $(\alpha, \gamma, t)$-graph families require that all the request rates are uniform. Here, we need to moderately round up the request rates for each object so that we have only a small number of distinct values of request rates.
\item \emph{Step 2. Decomposition.} We decompose the matching problem into smaller problems that can be addressed by Lemmas~\ref{lem:specialmatching} or~\ref{lem:fullblock}.
\end{itemize}

\para{Step 1. Rounding up.}
This is a standard trick: we construct a geometric sequence and round up each request rate to the nearest number in this sequence. Specifically, recall that $\epsilon$ is a suitably small constant. Let $\lambda = 1 - \epsilon/2$. Let $\lambda^{(1)} = (1-\epsilon)\widetilde T/2$. Let $\lambda^{(i)} = \lambda \cdot \lambda^{(i - 1)}$ for $2 \leq i \leq \frac{10}{\epsilon}\log n$. For each $p_j$, round it up to the smallest $\lambda^{(i)}/R$. We refer to the new rounded value as $\tilde p_i$.

One can check that $\sum_{i \leq k}\tilde p_i R \leq (1-c_0\epsilon)\alpha m \widetilde T$ for some constant $c_0$ and $\max_i \tilde p_i R \leq (1-\epsilon/3) \widetilde T/2$.
One can also see that if the rounded problem has a perfect matching, the original problem also has a perfect matching.
After the rounding, we have a total number of $K = \frac{10}{\epsilon}\log n$ types of request rates.

\begin{figure}[t]
\centering
\includegraphics[width=0.7\linewidth]{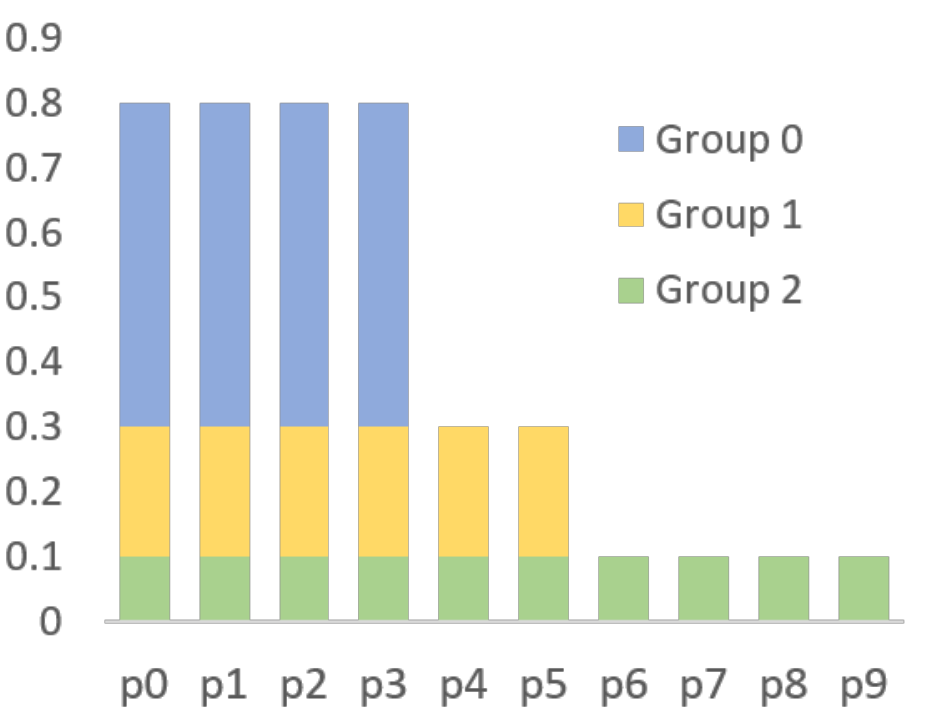}
\vspace{-0.1in}
\caption{An illustrative example of decomposing a long-tail request distribution $\{p_i\}_i$ into three homogeneous groups. Group 0 consists of a small number of objects that have high request rates. It can be handled by Lemma~\ref{lem:specialmatching}. Group 1 and 2 consist of a large number of requests that have low request rates. They can be handled by Lemma~\ref{lem:fullblock}.}
\vspace{-0.2in}
\label{fig:decomposition}
\end{figure}

\para{Step 2. Decomposition.}  We next decompose the matching problem on $G$ into smaller sub-problems. We start with a concrete example to illustrate our high-level idea. See Fig.~\ref{fig:decomposition}. In this example, we have a small portion of objects that have large request rates ($p_0$ to $p_3$) and a large portion of objects that have small request rates ($p_4$ to $p_9$). Now we may decompose the requests into three groups (represented by three colors). Group 0 (blue) consists of a small number of objects that have high request rates. It can be handled by Lemma~\ref{lem:specialmatching}. Group 1 and 2 (yellow and green) consist of a large number of requests that have low request rates. They can be handled by Lemma~\ref{lem:fullblock}. We can then use a union bound to analyze the probability that perfect matchings exist for all three groups.

We now formally explain our analysis. Specifically, let
$$L_i = \{i: \tilde p_i \geq \lambda^{(i)}/R\}.$$

Note that $L_1 \subseteq L_2 \subseteq \dots \subseteq L_K$. Let $i^*$ be the largest number such that $|L_{i^*}| \leq \alpha m$. Next, we define the following matching sub-problems:

\begin{itemize}
\item $\calM^{(0)}$: the object set is $L_{i^*}$. The request rate of a job $j \in
L_{i^*}$ is $\tilde p^{(0)}_j \leftarrow \tilde p_j - \lambda^{i^*+1}/R$ (when $i^*=K$,
set $\tilde p^{(0)}_j = \tilde p_j$). The processing rate of a cache node is $\widetilde T / 2$. This corresponds to Group 0 in Fig.~\ref{fig:decomposition}.
\item $\calM^{(i)}$ for $i > i^*$: the object set is $L_i$. The request rate of each object is $\delta^{(i)}= \lambda^{(i)}/R - \lambda^{(i + 1)}/R$.
The processing rate of a cache node is $\delta^{(i)}L_i/((1-\epsilon)m)$. This corresponds
to Groups 1 and 2 in Fig.~\ref{fig:decomposition}.
\end{itemize}

We note that we get $\tilde p_j$ if we sum up the rates for $o_j$ in all of the matching
sub-problems. Furthermore, if we sum up cache node $i$'s processing power in all the
matching sub-problem, it is $\widetilde T$. Therefore, if all of the matching sub-problems have perfect matching, then the original problem also has a perfect matching.

By Lemma~\ref{lem:specialmatching}, a perfect matching for $\calM^{(0)}$ exists with probability $2/m^{0.7}$. By Lemma~\ref{lem:fullblock}, a perfect matching for $\calM^{(i)}$ exists with probability $2/m^{0.7}$. By using a union bound, with probability $1 - 2K/m^{0.7} = 1 - o(1/m^{0.6})$ perfect matchings exist for each subproblem. Therefore, with high probability a perfect matching exists for the original $G$.
\end{proof}

\subsection{Feasible flows imply PoT is stationary (Proof of Lemma~2 in the paper)}\label{asec:flowimpliespot}
This section proves Lemma~2 in the paper. For completeness, we repeat the lemma below.

\begin{lemma}[Repeat of Lemma~2 in the paper]\label{lem:potstationary}
If a perfect matching exists for $G$, then the power-of-two-choices process is stationary.
\end{lemma}

Proving Lemma~\ref{lem:potstationary} requires us to use a powerful building block presented in~\cite{foss1998stability,foley2001join}.
Recall the set-up by using our notation.  Consider $2m$ exponential random variables with rate $\widetilde T_i > 0$. Each non-empty set of cache nodes $S \subseteq [2m]$, has an associated Poisson arrival process with rate $\lambda_S \geq 0$ that joins the shortest queue in $S$ with ties broken randomly. For each non-empty subset $Q \subseteq [2m]$, define the traffic intensity on $Q$ as

\begin{equation}\label{eqn:intensity}
\rho_Q = \frac{\sum_{S \subseteq Q}\lambda_S}{\mu_Q},
\end{equation}
where $\mu_Q = \sum_{i \in Q}\widetilde T_i$. Note that the total rate at which  objects  served by $Q$ can be greater than the numerator of (\ref{eqn:intensity}) since other requests may be allowed to be served by some or all of the cache nodes in $Q$. Let $\rho_{\max} = \max_{Q\subseteq[2m]}\{\rho_Q\}$. We have

\begin{theorem}\label{thm:intensity}~\cite{foss1998stability,foley2001join} Consider the above system. If $\rho_{\max} < 1$, then the Markov process is positive recurrent and has a stationary distribution $\pi$.
\end{theorem}

\begin{proof}[Proof of Lemma~\ref{lem:potstationary}] We need to describe our system in the language of Theorem~\ref{thm:intensity}.  Define $D(i) = \{a_{h_0(i)}, b_{h_1(i)}\}$. Let $S$ be an arbitrary subset of $\{A, B\}$. Define $\lambda_S$ as:
\begin{itemize}
\item If $S = \{a_i, b_j\}$ for some $i$ and $j$, let
\begin{equation}
\lambda_S = \sum_{i \leq k}(I(D(i) = S)) p_i R,
\end{equation}
where $I(\cdot)$ is an indicator function that sets to $1$ if and only if its argument is true.
\item Otherwise, $\lambda_S = 0$.
\end{itemize}
One can check that the above arrival process exactly describes our cache system. We next show that $\rho_{\max} < 1$.

Let $Q \subseteq A\cup B$. Let $J \subseteq [k]$ be the largest set such that for any $i \in J$, $\Gamma(i) \subseteq Q$. This implies $\Gamma(J) \subseteq Q$. On the other hand, because there is a perfect matching in $G$, we have $\sum_{i \in J}p_iR \leq (1-\epsilon)|\Gamma(J)|\widetilde T$. Therefore,
\begin{equation}
\rho_Q = \frac{\sum_{Q \subseteq A}\lambda_Q}{\mu_Q} = \frac{(1-\epsilon)\Gamma(|J|)\widetilde T}{\widetilde T |Q|} \leq (1-\epsilon).
\end{equation}
Therefore, $\rho_{\max} < 1$ and by Theorem~\ref{thm:intensity}, our process is stationary.
\end{proof}

Lemma~\ref{lem:matching} and Lemma~\ref{lem:potstationary} imply Theorem~\ref{the:balance}.
\begin{theorem}[Main Theorem]
Let $\alpha$ be a suitable constant.
If $k \leq m^\beta$ for some constant $\beta$ (i.e., $k$ and $m$ are polynomial-related)
and $\max_i(p_i) \cdot R \leq \widetilde{T}/2$, then for any $\epsilon > 0$,
the system is stationary for $R = (1-\epsilon)\alpha \cdot m\widetilde{T}$ and any $P$,
with high probability for sufficiently large $m$.
\label{the:balance}
\end{theorem}

\subsection{Proof of Lemma~\ref{lem:lower}}\label{asec:lower}
This section proves that when only one hash function is used, with constant probability
our system is not stationary. Without loss of generality, let $p_i R = 1$ for all
$i$, $\widetilde T > 1$ be an arbitrary constant, and $k = m$. Let $h:[k] \rightarrow A \cup B$ be the hash function we use. We continue to build a bipartite graph, in which the node set is $U \cup V$ ($U = \{o_i\}_{i \leq k}$ and $V = A \cup B$). Also, an edge $\{u, v\}$ is added if and only if $h(u) = v$. We shall show that with constant probability there exists an $S$ such that
\begin{equation}
|S| > \widetilde T |\Gamma(S)|.
\end{equation}

This means even when $\widetilde T$ is a constant times larger than $p_i R$, with constant probability requests from $S$ cannot be properly handled (i.e., the request rate is larger than the processing rate). We can use a standard anti-concentration result. Specifically we shall show that for a fixed $v$:
\begin{equation}\label{eqn:constantprob}
\Pr\left[\exists S: |S| = \widetilde T + 1 \wedge \left(h(s) = v \mbox{ for all } s \in S\right)\right] = \Omega(1)
\end{equation}
Define a sequence of indicator random variables $\{X_i\}_{i \leq k}$ such that $X_i = 1$ if and only if $h(i) = v$. Note that
{\small
\begin{equation}
\Pr\left[\sum_{i \leq k} X_i = \widetilde T + 1\right] = \binom{k}{\widetilde T + 1}\left(\frac 1 m\right)^{\tilde T + 1} \geq \left(\frac{1}{\widetilde T + 1}\right)^{\widetilde T + 1} = \Omega(1).
\end{equation}
}
This shows (\ref{eqn:constantprob}) and completes the proof of Lemma~\ref{lem:lower}.

\end{document}